\PassOptionsToPackage{svgnames}{xcolor}
\documentclass[11pt,letterpaper]{article}
\usepackage{fullpage}
\usepackage[square]{natbib}
\usepackage{mathpazo}
\usepackage{graphicx}
\usepackage{algorithm}
\usepackage[noend]{algorithmic}
\usepackage{amsmath,amsthm,amsfonts,bbm}
\usepackage[hidelinks,colorlinks,allcolors=DarkBlue,breaklinks]{hyperref}
\usepackage[nameinlink]{cleveref}
\usepackage{float}
\usepackage[dvipsnames]{xcolor}
\usepackage{tikz}
\usepackage{amsfonts}
\usepackage{nicefrac}
\usepackage{amsmath}
\usepackage{amssymb}
\usepackage{amsthm}
\usepackage{colortbl}
\usepackage{pdfpages}
\usepackage{thm-restate}

\usetikzlibrary{arrows.meta}
\usetikzlibrary{positioning}

\newtheorem{theorem}{Theorem}
\newtheorem{lemma}{Lemma}
\newtheorem{proposition}{Proposition}

\newtheorem{observation}{Observation}

\newtheorem{definition}{Definition}

\usepackage{todonotes}
\presetkeys{todonotes}{inline}{}

\renewcommand{\paragraph}[1]{\medskip\noindent\textbf{#1}~}

\usepackage{tcolorbox}
\newcounter{open}
\newenvironment{open}{\refstepcounter{open}\begin{tcolorbox}[colback=gray!10, colframe=black, rounded corners]\textbf{Open Question \theopen:}}{\end{tcolorbox}}

\renewcommand*{\le}{\leqslant}
\renewcommand*{\leq}{\leqslant}
\renewcommand*{\ge}{\geqslant}
\renewcommand*{\geq}{\geqslant}
\renewcommand{\epsilon}{\varepsilon}
\renewcommand{\succeq}{\succcurlyeq}
\renewcommand{\bar}{\overline}

\newcommand{\N}{\mathbb{N}}

\newcommand{\R}{\mathbb{R}}
\newcommand{\floor}[1]{\left\lfloor{#1}\right\rfloor}
\newcommand{\ceil}[1]{\left\lceil{#1}\right\rceil}
\newcommand{\set}[1]{\left\{{#1}\right\}}
\newcommand{\card}[1]{\left|{#1}\right|}

\newcommand{\sd}{\textrm{SD}}

\newif\ifacronyms
\acronymsfalse
\newlength{\abbrvskip}
\newcommand{\psd}{P-SD}
\newcommand{\pef}{P-EF}
\newcommand{\osd}{O-SD}
\newcommand{\oef}{O-EF}
\newcommand{\usd}{U-SD}
\newcommand{\uef}{U-EF}

\newcommand{\cspp}{\textsc{CutAndChoose++}}

\setcitestyle{authoryear}

\title{Temporal Fair Division}
\author{Benjamin Cookson \and Soroush Ebadian \and Nisarg Shah}
\date{%
    University of Toronto\\%
    \{bcookson,soroush,nisarg\}@cs.toronto.edu%}
}
\begin{document}

\maketitle

\begin{abstract}
We study temporal fair division, whereby a set of agents are allocated a (possibly different) set of goods on each day for a period of days. We study this setting, as well as a number of its special cases formed by the restrictions to two agents, same goods on each day, identical preferences, or combinations thereof, and chart out the landscape of achieving two types of fairness guarantees simultaneously: fairness on each day (per day) and fairness over time (up to each day, or the weaker version, overall).

In the most general setting, we prove that there always exists an allocation that is stochastically-dominant envy-free up to one good (SD-EF1) per day and proportional up to one good (PROP1) overall, and when all the agents have identical preferences, we show that SD-EF1 per day and SD-EF1 overall can be guaranteed. For the case of two agents, we prove that SD-EF1 per day and EF1 up to each day can be guaranteed using an envy balancing technique. We provide counterexamples for other combinations that establish our results as among the best guarantees possible, but also leaving open some tantalizing questions. 
\end{abstract}

\section{Introduction}\label{sec:intro}

How to divide a set of goods amongst a set of agents \emph{fairly} has been an enigma for centuries. There has been remarkable progress on this question in the last decade~\citep{ABFV22}. In the most prominent model, there is a set of $n$ agents $N$, each having an (additive) valuation over a set of goods $M$. The goal is to find an allocation $A = (A_1,\ldots,A_n)$ which partitions $M$ into pairwise-disjoint bundles, one allocated to each agent $i \in N$. 

This one-shot model fails to capture numerous real-world fair division scenarios in which goods are divided over time, e.g., food bank deliveries~\citep{lee2019webuildai}, resource allocation in data centers~\citep{ghodsi2013choosy}, allocation of advertising slots~\citep{MSVV07}, nurse shift scheduling~\citep{miller1976nurse}, and organ transplants~\citep{bertsimas2013fairness}. Compared to the one-shot setting, fair division of goods over time has received relatively little attention. 

Inspired by this, there has been a flurry of recent works that consider \emph{online fair division}, where agents or goods arrive over time and the principal needs to make allocations in an online fashion in the absence of any information regarding future arrivals~\citep{KPS14,benade2023fair}. The limit of feasible fairness guarantees have been explored under various adversary models~\citep{zeng2020fairness,benade2023fair}. 

However, in practice it is rarely the case that we have absolutely no information about the future. Significantly better guarantees have been established when even partial information about the future is available, either in the form of distributional knowledge~\citep{bogomolnaia2022fair,AKCM24} or machine-generated predictions~\citep{gkatzelis2021fair,banerjee2022online,BGHJ+23}.  But this work has left a very basic question wide open: \emph{How fair can we be if we had full information about the future?} 

To address this, we introduce the model of \emph{temporal fair division}, where a set of agents $N$ are allocated a set of goods $M_t$ on day $t$, over a period of days $t \in \set{1,\ldots,k}$, and the agents' valuations over the whole set of goods $M = \cup_{t=1}^k M_t$ are available upfront. At first glance, it may seem that this is just a traditional fair division problem where the set of goods $M$ needs to be divided amongst the set of agents $N$. The twist, however, is that in temporal fair division, agents anticipate fairness to prevail not solely at the end of the entire time horizon, but also at or within various interim time intervals. For example, the principal may be confident, based on their knowledge of the future, the allocation will eventually turn out to be fair, but that may not be assurance enough to the agents.

This leads us to seek \emph{temporal fairness notions} in our temporal fair division setting. Specifically, we take prominent fairness notions from one-shot fair division, and seek them on three temporal scales:
\begin{enumerate}
    \item[(1)] \emph{Per day:} The allocation of the set of goods $M_t$ on each day $t$ should be fair. 
    \item[(2a)] \emph{Overall:} The allocation of the whole set of goods $M$ in the end should be fair.
    \item[(2b)] \emph{Up to each day:} The allocation of the set of goods $\cup_{r=1}^t M_r$ up to each day $t$ should be fair. 
\end{enumerate}
Clearly, up to each day fairness (2b) is stronger than overall fairness (2a). Solely achieving per day fairness (1) or overall fairness (2a) can be reduced to one-shot fair division. Hence, we seek per day and overall fairness simultaneously (1+2a), or per day and up to each day fairness simultaneously (1+2b), or solely up to each day fairness (2b). Our main research question is to...
\begin{quote}
    \emph{...explore the limits of temporal fairness that can be guaranteed in temporal fair division.}    
\end{quote}

\subsection{Our Results \& Techniques}

\begin{table*}[ht]
\centering
\renewcommand{\arraystretch}{1.2}
\resizebox{\textwidth}{!}{%
\begin{tabular}{|l|l|l|l|l|}
\hline
& SD-EF1 up to each day & EF1 up to each day & SD-EF1(PROP1) Overall & EF1(PROP1) Overall \\
\hline
\textbf{General Setting} \\
\hline
SD-EF1 Per Day & X & X & X & \cellcolor{YellowGreen}$\checkmark^*$ (\Cref{thm:gen-PROP1-overall}) \\
EF1 Per Day & X & X & X & $\checkmark^*$ \\
$\emptyset$ & X & X~\citep{he2019informed} & \checkmark~\citep{aziz2020exanteexpost} & \checkmark~\citep{LMMS04} \\
\hline
\textbf{Two Agents} \\
\hline
SD-EF1 Per Day & X & \cellcolor{YellowGreen}\checkmark (\Cref{thm:2-sdef1per-ef1each})& X & \checkmark \\
EF1 Per Day & X & \checkmark & \cellcolor{YellowGreen}X (\Cref{thm:2-ef1per-sdef1over-NO}) & \checkmark \\
$\emptyset$ & \cellcolor{YellowGreen}X (\Cref{thm:2-sdef1upto-NO}) & \checkmark & \checkmark & \checkmark \\
\hline
\textbf{Identical Orderings} \\
\hline
SD-EF1 Per Day & X & ? & \cellcolor{YellowGreen}\checkmark (\Cref{thm:idpref-sdef1per-sdef1over}) & \checkmark \\
EF1 Per Day & X & ? & \checkmark & \checkmark \\
$\emptyset$ & \cellcolor{YellowGreen}X (\Cref{thm:idpref-iddays-sdef1uptoeach-NO}) & ? & \checkmark & \checkmark \\
\hline
\textbf{Identical Days} \\
\hline
SD-EF1 Per Day & X & ? & \cellcolor{YellowGreen}$\checkmark^*$ (\Cref{thm:iddays-SDPROP1-overall}) & $\checkmark^*$ \\
EF1 Per Day & X & ? & $\checkmark^*$ & $\checkmark^*$ \\
$\emptyset$ & \cellcolor{YellowGreen}X (\Cref{thm:idpref-iddays-sdef1uptoeach-NO}) & ? & \checkmark & \checkmark \\
\hline
\end{tabular}
}
%}
\caption{Possibilities, impossibilities, and open questions in temporal fair division. \checkmark indicates a possibility result; X indicates an impossibility. ? indicates an open question. $\checkmark^*$ indicates achieving (SD-)PROP1 in place of (SD-)EF1 (unless indicated by a $*$, all other results are for (SD-)EF1 rather than (SD-)PROP1). Green highlights indicate the main results of this paper; non-highlighted cells are either open, already known, or implied by other results.}
\label{table1}
\end{table*}

We chart out the landscape of the aforementioned temporal fair division model in a general setting, with $n$ agents having additive, heterogeneous preferences. Further, we identify three relevant restrictions where we can circumvent some of the impossibilities of the general setting, and achieve very strong results. Those are: (1) when there are only two agents; (2) when all agents have the same ordering over the goods; and (3) when an identical set of goods arrives each day.

\begin{figure}[htb!]
    \centering
    \begin{tikzpicture}
    \tikzset{
        mynode/.style={
            draw,
            rectangle,
            rounded corners=5pt,
            inner sep=5pt,
            minimum width=2.5cm
        },
        myarrow/.style={
            ->,
            thick,
            -{Latex[length=2mm]}
        },
        every node/.style={font=\small}
    }
    % Left side
    \node[mynode] (X1) at (-2,0) {SD-EF1 per day};
    \node[mynode] (X2) at (-2,-3) {EF1 per day};
    \draw[myarrow] (X1) -- (X2);
    \ifacronyms
    \node[left=\abbrvskip of X1]{\psd};
    \node[left=\abbrvskip of X2]{\pef};
    \fi
    
    % Right side
    \node[mynode] (Y1) at (2.8,0) {SD-EF1 up to each day};
    \node[mynode] (Y2) at (0.5,-1.5) {SD-EF1(PROP1) overall};
    \node[mynode] (Y3) at (5,-1.5) {EF1 up to each day};
    \node[mynode] (Y4) at (2.8,-3) {EF1(PROP1) overall};
    \draw[myarrow] (Y1) -- (Y2);
    \draw[myarrow] (Y1) -- (Y3);
    \draw[myarrow] (Y2) -- (Y4);
    \draw[myarrow] (Y3) -- (Y4);
    \ifacronyms
    \node[left=\abbrvskip of Y1]{\usd};
    \node[left=\abbrvskip of Y2]{\osd};
    \node[right=\abbrvskip of Y3]{\uef};
    \node[left=\abbrvskip of Y4]{\oef};
    \fi
    
    \end{tikzpicture}
    \caption{Hierarchy of temporal fairness notions.}
    \label{fig:hierarchy}
\end{figure}
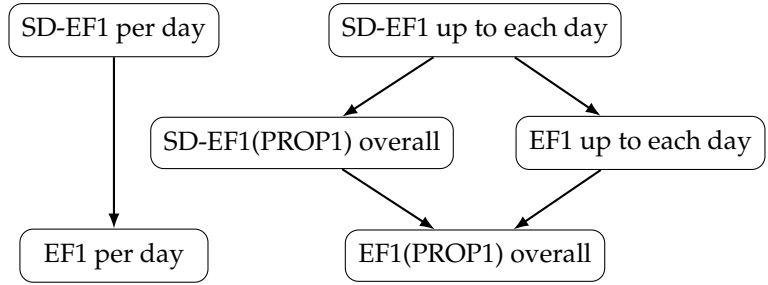

In these settings, we seek the fairness guarantees of EF1, SD-EF1, PROP1, and SD-PROP1 at the temporal scales of per day, overall, and up to each day. The various temporal fairness definitions are depicted in \Cref{fig:hierarchy}, with arrows indicating logical implications. We discover several surprising results, and develop novel algorithmic tools along the way, which may be of independent interest. 

In \Cref{sec:general}, we present an algorithm for finding temporally fair allocations in our most general setting. Specifically, we show how to obtain an allocation that is PROP1 overall, and SD-EF1 per day in polynomial time (\Cref{thm:gen-PROP1-overall}).

In \Cref{sec:two,sec:identical,sec:samedays}, we look at the restricted settings of two agents, identical orderings, and identical days, respectively. In these settings, we provide algorithms that give very strong fairness guarantees that are impossible in the general case (\Cref{thm:2-sdef1per-ef1each,thm:idpref-sdef1per-sdef1over,thm:iddays-SDPROP1-overall}), while also showing which fairness desiderata is still too strong even after applying these restrictions (\Cref{thm:2-ef1per-sdef1over-NO,thm:2-sdef1upto-NO,thm:idpref-iddays-sdef1uptoeach-NO}). For the two-agent case, we provide a complete picture of what is possible and impossible. In the case of identical orderings and identical days, we provide a near complete picture, while leaving open some interesting  questions surrounding the powerful notion of ``up to each day'' fairness. 

Finally, in \Cref{sec:laminar}, we consider broader interpretations of the temporal fairness model. We abstract the ideas of temporal fairness to explore achieving fairness over some set of goods $M$ while also achieving fairness over some collection of subsets of $M$. We talk about the implications of our results to this broader model, and provide extensions of some of our results.

Most of our results, along with several open questions, are summarized in \Cref{table1}.

\subsection{Related Work}\label{sec:related}
Our temporal fair division model is related to (but separate from) several fair division models studied in the literature. 

\paragraph{Repeated fair division.} The repeated fair division model of \citet{igarashi2024repeated}, which is the case of identical days in our more general model, is the most closely related to our work (though they deal with mixed goods and chores rather than just goods). Some of their results are for two agents (still with identical days). In \Cref{app:sec-4}, we show some results similar to their two agent results in our model, namely, that SD-EF1 per day and SD-EF1 up to each day can be achieved simultaneously, while also achieving SD-EF on every even day. This result is the only overlap between our works, and we present it again because we obtain a shorter proof with a much simpler algorithm when dealing with instances with only goods. The rest of their results with two or more agents seek exact fairness guarantees, such as (exact) envy-freeness overall, in limited cases such as when the number of days is a multiple of the number of agents. 

\paragraph{Online fair division.} In the online fair division model, goods arrive one by one and must be irrevocably allocated to an agent upon arrival with no knowledge of agent preferences over the goods to arrive later. Typically, one seeks to maintain a certain level of fairness. Clearly, any online fair division algorithm can be simulated in our temporal fair division model to achieve the same guarantee up to each day. One online fair division paper of particular note is that of \citet{he2019informed}. They introduce the ``Informed Model'' of online fair division, where irrevocable allocation decisions must be made as goods arrive one at a time in adversarial order, but the allocation algorithm is given all goods and the order they will arrive in upfront. The main goal in the informed model is to achieve an allocation that remains EF1 after each good is allocated. Clearly, this is equivalent to achieving EF1 up to each day in temporal fair division when each day only contains a single good. \citeauthor{he2019informed} conclude that it is impossible to allocate the goods in such a way that EF1 is always maintained. By corollary, EF1 up to each day is also infeasible. In another paper, \citet{benade2023fair} show that $O(\sqrt{k \log n})$ envy can be maintained up to $k$ days, and also point out that randomized algorithms may have much greater power against a nonadaptive adversary, who sets the full instance before the algorithm starts making random choices, with no super-constant envy lower bound known for this case. Online fair division with a nonadaptive adversary is still a stronger model than temporal fair division due to the fact that the algorithm does not have knowledge of what goods will arrive in future time periods. 

\paragraph{Repeated (or many-to-many) matching.} Our work is also inspired by the earlier work of \citet{gollapudi2020almost}, who consider the repeated two-sided matching problem, where there are $n$ agents on each side of a two-sided market with agents on each side having preferences over those on the other side, and the goal is to compute a perfect matching on each day over a period of days. They also seek guarantees such as EF1 up to each day. However, their positive results are only for binary valuations, and they leave achieving EF1 (for both sides) up to each day for general additive valuations as an open question. Finally, note that repeated perfect matching effectively produces a many-to-many matching. \citet{FMS21} study how to achieve EF1 (for both sides) in this setting, which can be viewed as an EF1 overall guarantee. They show how to achieve it when agents on each side have identical preferences, but leave it open for the case of general additive valuations. Note that unlike in fair division, EF1 overall is not straightforward in their case because EF1 needs to be achieved for agents on both sides simultaneously. 

\paragraph{Constrained fair division.} We remarked that achieving an overall fairness guarantee can be reduced to the one-shot fair division model, taking an instance with the set of all goods $M$. When we additionally seek a per day fairness guarantee, this can be modeled as a constraint on the space of feasible allocations, and the question becomes whether there is a \emph{constrained allocation} that still achieves the desired fairness guarantee. This model of constrained fair division has also been studied in the literature. \citet{biswas2018fair} study a model with cardinality constraints, where $M$ is partitioned into categories $(C_1,\ldots,C_p)$ and an allocation $A$ is feasible only if $|A_i \cap C_\ell| \le \ceil{|C_\ell|/n}$ for all $i,\ell$. That is, the allocation should divide the goods from each category as evenly as possible (in a ``balanced'' manner). As we remark in \Cref{sec:identical}, when agents have identical orderings over the goods, the SD-EF1 per day constraint can be reduced to a cardinality constraint, immediately yielding an allocation that is SD-EF1 per day and EF1 overall. However, for this case, we are able to achieve the stronger guarantee of SD-EF1 per day and SD-EF1 overall. Cardinality constraints (partition matroid constraints) have been generalized to matroid constraints, and the existence of an EF1 allocation subject to matroid feasibility constraints is a major open question~\citep{biswas2018fair,dror2023fair}. Finally, the bihierarchy framework of \citet{BCKM13} can also be viewed as a method for finding a constrained allocation, which we use in some of our results. Although, our most interesting results deal with sets of constraints that go beyond bihierarchies. 

Finally, contemporarily to and independently of our work, \citet{elkind2024temporal} also study the same temporal fair division model as ours, but their results have no overlap with ours. In particular, they focus only on EF1 up to each day, make the same observation as we do that the algorithm of \citet{he2019informed} effectively achieves this for allocating goods to two agents, and extend this in three ways: 
\begin{enumerate}
    \item They observe that the essentially same algorithm of \citet{he2019informed} extends to allocating chores to two agents as well. In comparison, we extend it to add stronger per day fairness guarantees and from ``up to each day'' to arbitrary laminar constraints.  
    \item They show that achieving both EF1 and PO up to each day is impossible, even with just two agents.
    \item They notice that for $n \ge 2$ agents and two identical days, performing a round-robin allocation (say in the order of agents $1$ through $n$) followed by a round-robin allocation in the reverse order (agents $n$ through $1$) achieves EF1 up to both days.\footnote{This is simply because no agent $i$ envies any agent $j$ with $j > i$ on day $1$, so ignoring the good picked by agent $j$ on day $2$ eliminates any envy from agent $i$ towards agent $j$, and a symmetric argument eliminates any envy from agent $j$ towards agent $i$ by ignoring the good picked by agent $i$ on day $1$.} As we observe, the existence of such an allocation for more than two days remains open.
\end{enumerate}
They also derive a number of hardness results for checking the existence of a temporally fair allocation when such existence is not guaranteed. 
%%%%%%%%%%%%%%%%%%%%%%%%%%%%%%%%%%%%%%%%%%%%%%%%%%%%%%%%%%
%%%%%%%%%%%%%%%%%%%%%%%%%%%%%%%%%%%%%%%%%%%%%%%%%%%%%%%%%%
%%%%%%%%%%%%%%%%%%%%%%%%%%%%%%%%%%%%%%%%%%%%%%%%%%%%%%%%%%

\section{Preliminaries}\label{sec:prelim}

\subsection{Model}
For any $r \in \N$, define $[r] \triangleq \{1,2,\dots,r\}$. A multiset is a set that allows repetitions.

\paragraph{Agents, goods, and valuations.} Let $N = [n]$ be a set of agents who are allocated a set of goods on each day over $k$ consecutive days. For $t \in [k]$, denote by $M_t$ the set of goods to be allocated on day $t$, $\bar{M}_t = \cup_{r \in [t]} M_r$ the set of goods up to day $t$, and $M = \bar{M}_k = \cup_{t \in [k]} M_t$ the set of all goods. We can view $(M_1,\ldots,M_k)$ as a partition of $M$. Each agent $i \in N$ has an additive valuation function $v_i : 2^M \to \R_{\ge 0}$, where $v_i(\set{g})$ (henceforth, with a slight abuse of notation, written as $v_i(g)$) is her utility for receiving good $g \in M$ and $v_i(S) = \sum_{g \in S} v_i(g)$ for all $S \subseteq M$. Collectively, $(N,(M_1,\ldots,M_k),\set{v_i}_{i \in N})$ form an instance of \emph{temporal fair division}. An instance with $k=1$ is a (regular) fair division instance, so a temporal fair division instance can be viewed as a sequence of fair division instances in which the same agents participate.

\paragraph{Preferences.} Define $\succeq_i$ (resp., $\succ_i$) as the weak (resp., strict) ordering over the goods in $M$ induced by $v_i$, where, for all $g,g' \in M$, $g \succeq_i g'$ if and only if $v_i(g) \geq v_i(g')$ and $g \succ_i g'$ if and only if $v_i(g) > v_i(g')$. For all $S \subseteq M$ and all $r \in \N$, define $T_i(S,r)$ to be the $r$ most preferred goods among the goods in $S$ according to the ordering $\succeq_i$; all ties are broken consistently across $i$, $S$, and $r$.\footnote{That is, we first use an arbitrary global ordering over $M$ as the tiebreaker to convert the weak ordering $\succeq_i$ of every agent $i \in N$ into her strict ordering over $M$, and then compute all $T_i(S,r)$-s according to these strict orderings. Our negative results do not depend on this tie-breaking and positive results hold regardless of it.} 

\paragraph{Allocations.} An allocation $A = (A_1,\ldots,A_n)$ is a partition of $M$ into $n$ pairwise-disjoint bundles, where $A_i$ is the bundle allocated to agent $i$. For $S \subseteq M$, let $A_S = (A_{S,1},\ldots,A_{S,n})$ be the allocation of the goods in $S$ that is induced by $A$ (i.e., for each good $g \in M$ and agent $i \in N$, $g \in A_{S,i}$ if and only if $g \in S$ and $g \in A_i$). For $t \in [k]$, we refer to $A_{M_t}$ as the allocation on day $t$ and $A_{\bar{M}_t}$ as the allocation up to day $t$.

\paragraph{Restrictions.} We study three restrictions of this general setup (and their combinations).

\begin{enumerate}
    \item \emph{Two agents:} $|N|=2$.
    \item \emph{Identical valuations/orderings:} Under identical valuations, $v_i = v$ for all agents $i \in N$. Under identical orderings, $\succeq_i = \succeq$ for all agents $i \in N$ (each agent has the same weak ordering over the goods). Here, we simply write $v$, $\succeq$, and $T(S,r)$, skipping the agent in the subscript. For this case, our results deal with desiderata which depend only on the orderings; thus, no distinction between valuations and orderings is necessary.\footnote{In other words, our positive results hold even under identical orderings (weaker restriction), while our negative results hold even under identical valuations (stronger restriction).}
    \item \emph{Identical days:} Informally, copies of the same goods are allocated on each day. Formally, for all days $t,t' \in [k]$, there is a bijection $f_{t,t'} : M_t \rightarrow M_{t'}$ such that $v_i(g) = v_i(f_{t,t'}(g))$ for all agents $i \in N$ and goods $g \in M_t$.
\end{enumerate}

\subsection{Fairness Desiderata}\label{sec:prelim-fairness}

Let us introduce envy-freeness, proportionality, and their variants, which are arguably the most appealing fairness desiderata and play a key role throughout our work. Later, we will introduce their temporal extensions. Other notions referred to in specific sections will be introduced therein.  

\begin{definition}[Envy-Freeness (EF)]
    An allocation $A$ of a set of goods $S$ is envy-free (EF) if for all $i,j \in N$, $v_i(A_i) \geq v_i(A_j)$, i.e., no agent strictly prefers another agent's allocation to her own.
\end{definition}

Unfortunately, EF cannot be guaranteed with indivisible goods (e.g., take the prototypical example of a single good and two agents). Hence, the following relaxation has been widely studied. 

\begin{definition}[Envy-Freeness Up to One Good (EF1)]
    An allocation $A$ of a set of goods $S$ is envy-free up to one good (EF1) if for all $i,j \in N$ with $A_j \neq \emptyset$, there exists a $g \in A_j$ such that $v_i(A_i) \geq v_i(A_j \setminus \set{g})$, i.e., no agent envies another agent if some good from the latter agent's bundle is removed.
\end{definition}

In addition to the above envy-based fairness notions, we will also introduce a weaker notion of measuring fairness that does not require directly comparing one agent's bundle to another's.

\begin{definition}[Proportionality (PROP)]
    An allocation $A$ of a set of goods $S$ is proportional (PROP), if for all $i \in N$, $v_i(A_i) \geq \frac{1}{n}v_i(S)$, i.e., every agent gets a bundle that they view as being worth at least a $\frac{1}{n}$ share of the total set of goods $S$.
\end{definition}

However, just like EF, PROP is too strong to be universally guaranteed (the same counterexample as above will suffice to show this). Therefore, we also introduce an analogous weakening of PROP.

\begin{definition}[Proportionality Up to One Good (PROP1)]
    An allocation $A$ of a set of goods $S$ is proportional up to one good (PROP1) if for all $i \in N$ such that $A_i \not= S$, there exists a good $g \in S \setminus A_i$, such that $v_i(A_i \cup \set{g}) \geq \frac{1}{n}v_i(S)$.
\end{definition}

It is well known that EF1 is a stronger notion than PROP1 \citep{CFS17}.

Given a set of goods $S$ and a good $g \in S$, define an agent $i$'s top-set with respect to $S$ as $H_i(S,g) = \set{g' \in S : g' \succeq_i g}$. When given only a weak ordering $\succeq_i$ over a set of goods $S$, we can compare two bundles $X,Y \subseteq S$ using the \emph{stochastic dominance} (SD) relation: $X \succeq_i^\sd Y$ if for all $g \in S$, $\card{X \cap H_i(S,g)} \ge \card{Y \cap H_i(S,g)}$. That is, $X$ has at least as many goods weakly preferred to any good as $Y$ has. It is known that $X \succeq_i^\sd Y$ if and only if $v_i(X) \ge v_i(Y)$ for all (additive) valuations $v_i$ over $S$ that would induce $\succeq_i$. Hence, using the SD comparison in the EF and EF1 definitions yields their stronger counterparts, which have also been studied extensively~\citep{AGMW15,FMS21,AFSV23}.

\begin{definition}[SD-EF]
    An allocation $A$ of a set of goods $S$ is stochastically-dominant envy-free (SD-EF) if for all $i,j \in N$, $A_i \succeq_i^\sd A_j$.
\end{definition}

\begin{definition}[SD-EF1]
    An allocation $A$ of a set of goods $S$ is stochastically-dominant envy-free up to one good (SD-EF1) if for all $i,j \in N$ with $A_j \neq \emptyset$, there exists a $g \in A_j$ such that $A_i \succeq_i^\sd A_j \setminus \set{g}$.
\end{definition}

We also introduce a stochastic-dominance extension of PROP1. 
\begin{definition}[SD-PROP1]
    An allocation $A$ of a set of goods $S$ is stochastically-dominant proportional up to one good (SD-PROP1) if, for all $i \in N$ with $A_i \neq S$, there exists a $g \in S \setminus A_i$ such that $|(A_i \cup \set{g}) \cap H_i(S,g')| \ge \ceil{\nicefrac{|H_i(S,g')|}{n}}$ for all $g' \in S$.
\end{definition}
Informally, $A_i$, after adding at most one good to it, must contain at least $\ceil{\nicefrac{k}{n}}$ goods among the $k$ most preferred goods of agent $i$ in $S$, for each $k \in [|S|]$.

Just as EF1 implies PROP1, we have that SD-EF1 implies SD-PROP1, and similarly to SD-EF1, if an allocation $A$ is SD-PROP1 for certain orderings $\set{\succeq_i}_{i \in N}$, then $A$ will be PROP1 for any additive valuation functions that induce $\set{\succeq_i}_{i \in N}$. Both these facts are proven in \Cref{app:sec-2}.

Finally, we establish some necessary and sufficient conditions for the existence of SD-EF1 allocations. These technical results help simplify many of the proofs that rely on stochastic-dominance, as they relate SD-EF1 to the $T_i(S,g)$ function defined earlier.

\begin{restatable}{proposition}{sdefprop}\label{prop:sd-ef1}
    Let $A$ be an allocation of a set of goods $S$. 
    \begin{itemize}
        \item (Sufficiency) If $\card{A_i \cap T_i(S,r)} \ge \card{A_j \cap T_i(S,r)} - 1$ for all $i \in N$ and $r \in [|S|]$, then $A$ is SD-EF1. If $n=2$, the condition can be written as $\card{A_i \cap T_i(S,r)} \ge \floor{r/n}$. If $T_i(S,r) = T_j(S,r)$ for all $i,j \in N$ and $r \in [|S|]$, the condition can be written as $\card{A_i \cap T_i(S,r)} \in \set{\floor{r/n},\ceil{r/n}}$.
        \item (Necessity) If $A$ is SD-EF1, then $\card{A_i \cap T_i(S,r)} \ge \floor{r/n}$ for all $i \in N$ and $r \in [|S|]$ conditioned on $g \succ_i g'$ for all $g \in T_i(S,r)$ and $g' \in S\setminus T_i(S,r)$. Further, if $T_i(S,r) = T_j(S,r)$ for all $i,j \in N$ and $r \in [|S|]$, then the condition can be written as $\card{A_i \cap T_i(S,r)} \in \set{\floor{r/n},\ceil{r/n}}$. 
    \end{itemize}
\end{restatable}

The proof of \Cref{prop:sd-ef1}, along with some additional discussion on its usefulness is provided in \Cref{app:sec-2}.    

\subsection{Temporal Fairness}\label{sec:prelim-temporal}

In a temporal fair division instance given by a set of goods $M$ partitioned as $(M_1,\ldots,M_k)$ across $k$ days, we can ask for fairness to hold at different levels of granularity, yielding various temporal extensions of the fairness desiderata introduced above. These extensions also apply to any other type of desiderata (e.g., efficiency). 

\begin{definition}[Per Day Fairness]
    For desideratum $X$, allocation $A$ satisfies $X$ \textit{per day} if $A_{M_t}$ satisfies $X$ for all $t \in [k]$. 
\end{definition}

\begin{definition}[Overall Fairness]
    For desideratum $X$, allocation $A$ satisfies $X$ \emph{overall} if $A_M = A$ satisfies $X$.
\end{definition}

\begin{definition}[Up To Each Day Fairness]
    For desideratum $X$, allocation $A$ satisfies $X$ \textit{up to each day} if $A_{\bar{M}_t}$ satisfies $X$ for all $t \in [k]$.
\end{definition}

Note that `up to each day' is a strengthening of `overall', while `per day' is incomparable to those two. Plugging in our fairness desiderata into these three temporal extensions gives us the hierarchy of fairness guarantees depicted in \Cref{fig:hierarchy}. Because SD-EF1 is achievable for (regular) fair division (e.g., via a simple round-robin procedure~\citep{CKMP+19}), SD-EF1 per day and SD-EF1 overall are both individually achievable, implying the same for EF1, PROP1, and SD-PROP1. 

%%%%%%%%%%%%%%%%%%%%%%%%%%%%%%%%%%%%%%%%%%%%%%%%%%%%%%%%%%
%%%%%%%%%%%%%%%%%%%%%%%%%%%%%%%%%%%%%%%%%%%%%%%%%%%%%%%%%%
%%%%%%%%%%%%%%%%%%%%%%%%%%%%%%%%%%%%%%%%%%%%%%%%%%%%%%%%%%

\section{General Preferences}\label{sec:general}

In this section, we present temporal fair division results in the most general setting: an arbitrary set of goods arrives each day, and each agent has arbitrary additive preferences over them. Let us present our main result for this general setting. 

\begin{theorem}\label{thm:gen-PROP1-overall}
    For any temporal fair division instance, an allocation that is SD-EF1 per day and PROP1 overall exists and can be computed in polynomial time.
\end{theorem}

We find such an allocation using \Cref{alg:one}. The derivation of \Cref{thm:gen-PROP1-overall} can be divided into three conceptual steps. 

\newcommand{\ccalg}{\operatorname{BarmanBiswasCC}}

\begin{algorithm}[t]
\textbf{Input} {A Temporal Fair Division Instance $(N,M=\set{M_1,\dots,M_k},v)$}
    
\textbf{Output} {An allocation $A$ of the set of goods $M = \cup_{t \in [k]} M_t$}
\caption{SD-EF1 Per Day + PROP1 Overall}\label{alg:one}
\begin{algorithmic}[1]
\COMMENT{\small // Identical ordering transformation}
\STATE $v'_i = \emptyset$ for all $i \in N$
\FOR{$t \in [k]$}
    \STATE $M'_t \gets \{g'_{t, 1}, \ldots, g'_{t, |M_t|}\}$
    \FOR{$i \in N$} 
        \STATE $o_{i, t}$ be the goods $M_t$ in non-increasing order of $v_i$
        \STATE $v'_i \gets v'_i \cup \{g'_{t, j} \to v_i(o_{i, t}(j))\}$ for all $j \in [|M_t|]$
    \ENDFOR
\ENDFOR

\vspace{0.5em}
\COMMENT{\small // Invoking the result on EF1 subject to cardinality constraints}
\FOR{$t \in [k]$}
    \STATE Divide $M'_t$ into groups of size $n$ (last one may have less than $n$), i.e., $C_{t, 1} \gets \{g'_1, \ldots, g'_n\}$, $C_{t, 2} \gets \{g'_{n + 1}, \ldots, g'_{2n}\},$ $\ldots$, $C_{t, \lceil |M'_t| / n \rceil} \gets \{g'_{k'n}, \ldots g'_{|M'_t|}\}$
\ENDFOR
\STATE $A' \gets \ccalg(\bigcup_{t \in [k]} \bigcup_{j} C_{t, j}, v')$

\vspace{0.5em}
\COMMENT{\small // Final allocation with daily picking sequences based on $A'$}
\STATE $A \gets \emptyset$
\FOR{$t \in [k]$}
    \FOR{$j \in [|M_t|]$}
        \STATE $i \gets $ Agent allocated $g_{t,j}'$ in $A'$
        \STATE $g \gets i$'s favourite unallocated good from $M_t$
        \STATE $A_i \gets A_i \cup \set{g}$
    \ENDFOR
\ENDFOR
\RETURN $A$
\end{algorithmic}
\end{algorithm}

\paragraph{Identical ordering transformation.} \Cref{alg:one} begins by creating an auxiliary temporal fair division instance as follows. For each day $t$, it creates a new instance for that day with a set of goods $M'_t$ and valuations $v'$ such that agents have identical orderings but with the same set of utility values as they had previously. More formally, let $M'_t = \{g'_{t,1}, \ldots, g'_{t, |M'_t|}\}$. Then, for each agent $i \in N$, $v'_i(g'_{t, 1}) \ge v'_i(g'_{t, 2}) \ldots \ge v'_i(g'_{t, |M'_t|})$ (common ordering) and there exists a bijection $o_{i, t}$ between $M_t$ and $M'_t$ such that $v_i(g) = v'_i(o_{i, t}(g))$ for all $g \in M_t$ (same utility values). This technique has been used previously for designing algorithms to achieve (approximate) maximin share fairness (MMS)~\citep{BL14}, but we use it with a novel and nontrivial analysis to ensure PROP1.

\paragraph{Connection to cardinality constraints.}
\Cref{alg:one} invokes a key subroutine due to \citet{biswas2018fair} that returns an EF1 allocation subject to cardinality constraints summarized below.

\begin{theorem}[Theorem 1 of \citep{biswas2018fair}]
\label{thm:cc}
Given $p$ disjoint sets of goods $C_1, \ldots, C_p$ and $n$ agents with heterogeneous additive valuations, there always exists an EF1 allocation $A$ such that $\lfloor |C_\ell|/n \rfloor \le |A_i \cap C_\ell| \le \lceil |C_\ell|/n \rceil$ for every agent $i \in N$ and $\ell \in [p]$, and such an allocation can be computed in polynomial time.
\end{theorem}

\Cref{alg:one} invokes the algorithm of \Cref{thm:cc} on the following instance. Fix a day $t$. Recall that agents have identical orderings for $M'_t$ given as $g'_{t,1} \succeq_i \ldots \succeq_i g'_{t,|M'_t|}$ for all $i \in N$. Divide $M'_t$ into groups of size $n$ in the decreasing order of value, breaking ties arbitrarily and letting the last group have possibly fewer than $n$ goods: that is, let $C_{t, 1} = \{g'_{t,1}, \ldots, g'_{t,n}\}$, $C_{t, 2} = \{g'_{t, n + 1}, \ldots, g'_{t, 2n}\}$, and so on. By \Cref{thm:cc}, we find an allocation $A'$ that is an EF1 allocation for $\bigcup_{t \in [k]} M'_t$ and $v'$, and that each agent is allocated at most one good from $C_{t, j}$ for all $t$ and $j$.

\paragraph{Final Allocation.} \Cref{alg:one} then takes the allocation $A'$ and, for each day $M'_t = \{g'_{t,1} \ldots, g'_{t, |M'_t|}\}$, allocates $M_t$ 
through a ``serial dictatorship'' with the picking sequence derived from $A'$. First, the agent that is allocated $g'_{t, 1}$ in $A'$ will pick their favourite good from $M_t$ (the original set of goods for day $t$); next, the owner of $g'_{t, 2}$ picks their favourite good among the remaining goods of $M_t$; and so on. We now prove that the resulting allocation $A$ is SD-EF1 per day and PROP1 overall.

\begin{lemma}\label{lem:gen-alg1-sdef1}
\Cref{alg:one} returns an allocation that is SD-EF1 per day.
\end{lemma}
\begin{proof}
    Let $A$ be the allocation returned by \Cref{alg:one}. Due to \Cref{prop:sd-ef1}, to prove that $A$ is SD-EF1 per-day, it is sufficient to show that $\card{A_{M_t,i} \cap T_i(M_t,r)} \ge \card{A_{M_t,j} \cap T_i(M_t,r)} - 1$ for all $i,j \in N$, $t \in [k]$, and $r \in [|M_t|]$.

    Let $A'$ represent the allocation that the algorithm finds over the identical ordering transformation instance $M'$. In each day $t \in [k]$, we have that $M'_t$ is partitioned into sets $C_{t, 1} = \{g'_{t,1}, \ldots, g'_{t,n}\}$, $C_{t, 2} = \{g'_{t, n + 1}, \ldots, g'_{t, 2n}\}$, and so on, in order of utility value. Without loss of generality, assume that the labeling of the goods is also consistent with the tie-break ordering of the $T$ function. The allocation $A'_{M'_t}$ will guarantee that no agent receives more than $1$ good from each set $C_{t,l}$ for all $l$. Note that this will imply that for all $i \in N$ and $r \in [|M'_t|]$, we will have $|A'_{M'_t,i} \cap T(M'_t,r)| \in \set{\floor{r/n},\ceil{r/n}}$. By the sufficiency condition of \Cref{prop:sd-ef1}, this means that $A'_{M'_t}$ will be SD-EF1.

    Next, \Cref{alg:one} uses a picking order procedure to construct the final allocation $A_{M_t}$ from $A'_{M'_t}$. It can be seen that each agent is assigned exactly one ``pick'' over the goods in $M_t$ for each good they received in the allocation $A'_{M'_t}$, with the ordering of these picks corresponding to the common preference ordering over $M'_t$. Therefore, we know that for any $r \in [|M_t|]$, after the $r$-th pick of the procedure, each agent will have received either $\floor{r/n}$ or $\ceil{r/n}$ picks. It can also be seen that after the $r$-th pick of the procedure, $\card{A_{M_t,i} \cap T_i(M_t,r)} \geq \floor{r/n}$ will be true for each agent $i$. This is because at each pick $r$, the picking agent will select their most preferred good from $M_t$ that has not yet been picked. After pick $r$, only $r$ goods from $M_t$ have been assigned, and $\card{T_i(M_t,r)} = r$, so each of an agent's picks up to and including the $r$-th pick of the procedure will all have been used to select an item from their top $r$ goods from $M_t$. 

    For contradiction, assume our original claim is not true, the allocation on some day is not SD-EF1, and therefore, there exists some agents $i,j \in N$, some day $t \in [k]$, and some $r \in [|M_t|]$, such that $\card{A_{M_t,i} \cap T_i(M_t,r)} < \card{A_{M_t,j} \cap T_i(M_t,r)} - 1$.

    Let $r'$ be the last pick in the picking sequence where agent $j$ picked a good from $T_i(M_t,r)$. We know that $\card{A'_{M'_t,i} \cap T(M'_t,r')} \in \set{\floor{r'/n},\ceil{r'/n}}$, so agent $i$ must have received at least $\floor{r'/n}$ picks prior to pick $r'$, and each of those picks must have been used to select a good from $T_i(M_t,r)$ (Agent $i$ would never have used one of these picks to select a good not from $T_i(M_t,r)$, since we know there was at least one good from $T_i(M_t,r)$ available, the good that agent $j$ selected with pick $r'$). Similarly, we know that $\card{A'_{M'_t,j} \cap T(M'_t,r')} \in \set{\floor{r'/n},\ceil{r'/n}}$, so agent $j$ could only have had a maximum of $\ceil{r'/n}$ picks up to and including the $r'$-th overall pick. Since $r'$ is the last pick where agent $j$ selected a good from $T_i(M_t,r)$, that means $A_j$ can only contain at most $\ceil{r'/n}$ goods from $T_i(M_t,r)$. This gives us a contradiction since $\floor{r'/n} \geq \ceil{r'/n} - 1$.
\end{proof}

Intuitively, it is easy to see that $A_{M_t}$ will be EF1 for each day $t$, due to the way the allocation $A'$ is constructed, it can be seen that the picking order will be ``recursively-balanced'', which is well known to yield SD-EF1~\citep{aziz2020exanteexpost}.

To prove that $A$ is PROP1 overall, we use the following technical lemma.

\begin{lemma}
\label{lem:prop1-helper}
Let $V$ be a multiset of $m$ real values, and $A = \set{a_1,\ldots, a_k}$ and $A' = \set{a'_1, \ldots, a'_k}$ be two subsets of $V$ with equal size such that $a_j \ge a'_j$ for all $j \in [k]$. Let $g = \max\{x : x \in V \setminus A\}$ and $g' = \max\{x : x \in V \setminus A'\}$. Then, there exists a bijection $z$ from $A \cup \{g\}$ to $A' \cup \{g'\}$ such that $x \ge z(x)$ for all $x \in A \cup \{g\}$.
\end{lemma}
\begin{proof}
    In the case where $g \ge g'$, then $z$ can be constructed by simply mapping $g$ to $g'$ and $a_j$ to $a'_j$ for all $j \in [k]$.

    Now, the case where $g < g'$. Suppose $g$ is the $r$th largest element of $V$ (ties broken such that $g$ is ranked lowest possible). Since $g$ is the maximum among $V \setminus A$, the first $r-1$ elements of $V$ are in $A$. Therefore, $\{a_1, \ldots, a_{r - 1}\} \cup \{g\}$ are the top $r$ elements of $V$, and we can simply create a bijection from those elements to $\{a'_1, \ldots a'_{r - 1}\} \cup \{g'\}$ as desired. For the remaining $k - r + 1$ bottom ranks of both $A$ and $A'$, we simply use $z(a_j) = a'_j$ since we know from the statement that $a_j \ge a'_j$, which completes the proof.
\end{proof}

\begin{lemma}
\label{lem:alg1-prop1}
\Cref{alg:one} returns an allocation that is PROP1 overall.
\end{lemma}

\begin{proof}
Fix an agent $i \in N$. Take a day $t \in [k]$. Rename the goods so that $A'_i \cap M'_t = \{g'_1, \ldots, g'_{|A'_i \cap M'_t|}\}$ are the goods that $i$ is allocated in $A'$ in a non-increasing order of $v'_i$. Similarly, let $A_i \cap M_t = \{g_1, \ldots, g_{|A_i \cap M_t|}\}$ be the goods $i$ picks according to the picking sequence in order. That is, $g_1$ is the good picked corresponding to $g'_1$, $g_2$ corresponding to $g'_2$, and so on. 

Towards invoking \Cref{lem:prop1-helper}, a helpful observation is that
$v_i(g_j) \ge v'_i(g'_j)$
for all $j \in [|A_i \cap M_t|]$. 
Suppose $g'_j$ is the $r$th preferred good among $M'_t$ for $i$. Since $M_t$ and $M'_t$ share the same multiset of utility values and $g_j$ is the top pick  of $i$ when $r - 1$ goods are picked, $g_j$ is at least as good as the $r$-th good among $M_t$ (and hence, $M'_t$). This argument, combined across all days, implies existence of a bijection $z_i: A_i \to A'_i$ such that $v_i(g) \ge v'_i(z_i(g))$ for all $g \in A_i$.

By invoking \Cref{lem:prop1-helper} with $A \gets A_i$ and $A' \gets A'_i$ over the multiset $V$ being $i$'s utility values, we have that 
\[
v_i(A_i) + \max_{g \notin A_i} v_i(g) \ge 
v'_i(A'_i) + \max_{g \notin A'_i} v'_i(g) 
\]

Every EF1 allocation is also PROP1 \citep{CFS17}, therefore, since $A'$ is EF1 (\Cref{thm:cc}), we have that
\[
v'_i(A'_i) + \max_{g \notin A'_i} v'_i(g) \ge \frac{1}{n} v_i(M') = \frac{1}{n} v_i(M),
\]
the last equality being true from the way we constructed $M'$. Combining the two inequalities above, we get $v_i(A_i) + \max_{g \notin A_i} v_i(g) \ge \frac{1}{n}v_i(M)$. Thus, $A$ is PROP1.
\end{proof}

It is worth noting that PROP1 is not a monotonic property, i.e., if $v_i(A_i) \ge v_i(A'_i)$ and $A'_i$ is PROP1, it is possible that $A_i$ is not PROP1 (as the best good for $i$ in $M \setminus A_i$ could be worth less than the best good in $M \setminus A'_i$). This is why we needed to use a more involved argument in \Cref{lem:prop1-helper,lem:alg1-prop1}.

%%%%%%%%%%%%%%%%%%%%%%%%%%%%%%%%%%%%%%%%%%%%%%%%%%%%%%%%%%
%%%%%%%%%%%%%%%%%%%%%%%%%%%%%%%%%%%%%%%%%%%%%%%%%%%%%%%%%%
%%%%%%%%%%%%%%%%%%%%%%%%%%%%%%%%%%%%%%%%%%%%%%%%%%%%%%%%%%

\section{Two Agents}\label{sec:two}

In this section, we consider temporal fair division with two agents. We provide a complete picture of temporal fairness notions that can be guaranteed in this case. We establish a strong positive result, then show that it is the best possible by producing counterexamples for stronger desiderata.

\subsection{Possibilities}\label{sec:two-positive} 

Our main goal in this section is to show that SD-EF1 per day and EF1 up to each day can be achieved for two agents. We begin by introducing an \emph{envy-balancing lemma}, a powerful tool for finding temporal allocations to two agents. We later use this lemma to derive not only the aforementioned guarantee, but also other appealing guarantees.

\begin{definition}[Cancelling Allocations]
    We say that allocations $B$ and $B'$ of a set of goods $S$ to two agents \emph{cancel out} if 
    $v_i(B_i)+v_i(B'_i) \ge v_i(B_{3-i})+v_i(B'_{3-i}), \forall i \in [2].$
    In words, they cancel out if hypothetically allocating two copies of each good in $S$, one according to $B$ and the other according to $B'$, achieves (exact) envy-freeness.
\end{definition}

\begin{algorithm}
    \caption{Envy-Balancing Algorithm}\label{alg:envy-balancing}
    \textbf{Input} {A Temporal Fair Division Instance $(N,M=\set{M_1,\dots,M_k},v)$; For $t \in [k]$, a pair of allocations $(B^1_t,B^2_t)$ of the set of goods $M_t$ that cancel out, with labels assigned in such a way that if neither of the allocations are envy-free, then $v_1(B^1_{t,1}) - v_1(B^1_{t,2}) \geq 0$ and $v_2(B^2_{t,2}) - v_2(B^2_{t,1}) \geq 0$.}
    
	\textbf{Output} {An allocation $A$ of the set of goods $M = \cup_{t \in [k]} M_t$}
    \begin{algorithmic}[1]
        \STATE $F \gets \emptyset$, $S \gets \emptyset$, $e_1 \gets 0$, $e_2 \gets 0$
        \FOR{$t \in [k]$}
            \IF{$v_1(B^1_{t,1}) \geq v_1(B^1_{t,2}) \land v_2(B^1_{t,2}) \geq v_2(B^1_{t,1})$}
                \COMMENT {\small // If $B^1_t$ is EF}
                \STATE $F \leftarrow F \cup \{B^1_t\}$
            \ELSIF{$v_1(B^2_{t,1}) \geq v_1(B^2_{t,2}) \land v_2(B^2_{t,2}) \geq v_2(B^2_{t,1})$}
                \COMMENT {\small // If $B^2_t$ is EF}
                \STATE $F \leftarrow F \cup \{B^2_t\}$
            \ELSE
                \IF{$e_1 \leq 0$}
                    \COMMENT {\small // Agent $1$ feels envy in $A_S$, or $S$ is empty}
                    \STATE $S \leftarrow S \cup \{B^1_t\}$
                    \STATE $e_1 \leftarrow e_1 + (v_1(B^1_{t,1}) - v_1(B^1_{t,2}))$
                    \STATE $e_2 \leftarrow e_2 + (v_2(B^1_{t,2}) - v_2(B^1_{t,1}))$
                \ELSE
                    \COMMENT {\small // If Agent $2$ feels envy in $A_S$}
                    \STATE $S \leftarrow S \cup \{B^2_t\}$
                    \STATE $e_1 \leftarrow e_1 + (v_1(B^2_{t,1}) - v_1(B^2_{t,2}))$
                    \STATE $e_2 \leftarrow e_2 + (v_2(B^2_{t,2}) - v_2(B^2_{t,1}))$
                \ENDIF
            \ENDIF

            \IF{$e_1 \geq 0 \land e_2 \geq 0$}
                \COMMENT {\small // If $A_S$ is EF}
                \STATE $F \leftarrow F \cup S$
                \STATE $S \leftarrow \emptyset$, $e_1 \leftarrow 0$, $e_2 \leftarrow 0$
            \ELSIF{$e_1 \leq 0 \land e_2 \leq 0$}
                \COMMENT {\small // If both agents either feel envy, or are indifferent in $A_S$}
                \STATE $F \leftarrow F \cup \text{SWAP}(S)$
                \STATE $S \leftarrow \emptyset$, $e_1 \leftarrow 0$, $e_2 \leftarrow 0$
            \ENDIF
        \ENDFOR
        \STATE $F \leftarrow F \cup S$
        \STATE $A \gets$ allocation in which $M_t$ is allocated according to the allocation of $M_t$ in $F$, for each $t \in [k]$
        \RETURN $A$
    \end{algorithmic}
\end{algorithm}

\begin{lemma}[Envy-Balancing Lemma]\label{lem:envy-balancing}
    Suppose that for each day $t \in [k]$, we are given two EF1 allocations $B_t$ and $B'_t$ of the set of goods $M_t$ that cancel out. Then, we can compute, in polynomial time, an allocation $A$ of the set of all goods $M = \cup_{t \in [k]} M_t$ that is EF1 up to each day and $A_{M_t} \in \set{B_t,B'_t}$ for each day $t \in [k]$. 
\end{lemma}

For readers familiar with the ``informed'' model of online fair division from \citet{he2019informed}, this can be seen as a generalization of their two agent algorithm. Their algorithm shows how to achieve EF1 up to each day when allocating a single good during each time step. We provide a similar guarantee while allocating batches of goods and simultaneously maintaining fairness over the batches. 

We prove \Cref{lem:envy-balancing} from scratch, without invoking the algorithm of \citet{he2019informed}, but one may observe that it could also be proved by reducing our setting to that of \citet{he2019informed}. Specifically, on every day $t \in [k]$, when given two allocations $B_t$ and $B'_t$ over $M_t$ that cancel out, we can reduce the problem of deciding whether the goods in $M_t$ should be allocated according to $B_t$ or according to $B'_t$ to deciding whether a single good should be allocated to agent $1$ or to agent $2$ in the reduced setting. The key to this reduction is that whenever one of the allocations $B_t$ or $B'_t$ would cause some agent $i$ to feel envy on day $t$, we can set agent $i$'s valuation for the single good on day $t$ in the reduced setting to be the negation of the envy they feel when that unpreferred allocation is chosen; when this single good is given to the other agent $3-i$ in the reduced setting, it will accurately represent the envy of agent $i$. Due to the fact that $B_t$ and $B'_t$ cancel out, allocating the single good to agent $i$ in the reduced setting will only ever underestimate the positive marginal utility they have for their own bundle in the preferred allocation over that of the other agent, meaning that under the actual utilities everyone will be better off compared to the reduced setting.

We choose to prove \Cref{lem:envy-balancing} directly rather than through this reduction as it provides a more straightforward view of the new fairness guarantees of our algorithm, and better deals with some of the subtleties present in the reduction, such as showing that EF1 up to each day in the reduced setting still translates to EF1 up to each day in the original setting when the meaning ``one good'' changes between the two settings. We also later use this proof as the basis of our proof of \Cref{lem:laminar-envy-balancing}, an extension of the Envy-Balancing Lemma which does not follow from the results of \citet{he2019informed}.

\begin{proof}[Proof of \Cref{lem:envy-balancing}]
     For any allocation $A$, we will use the following language to describe the agents' relative valuations of their bundles compared to the other agent:
    \begin{itemize}
        \item For any allocation $A$ where $v_i(A_{3-i}) > v_i(A_i)$ for some agent $i \in [2]$, the negative value $v_i(A_{i}) - v_i(A_{3-i})$ will be referred to as the ``envy'' felt by agent $i$ in $A$.
        \item Similarly, for any allocation $A$ where $v_i(A_i) > v_i(A_{3-i})$ for some agent $i \in [2]$, the positive value $v_i(A_{i}) - v_i(A_{3-i})$ will be referred to as the ``surplus utility'' felt by agent $i$ in $A$.
    \end{itemize}

    For each $t \in [k]$, let $B^1_t, B^2_t$ be two allocations over $M_t$ that are both EF1 and cancel out. Since these allocations cancel out, we can make the following assumption without loss of generality:

    If neither of $B^1_t$ or $B^2_t$ are EF, then we have the following four inequalities:
    \begin{align*}
    &v_1(B^1_{t,1}) > v_1(B^1_{t,2}), &v_2(B^1_{t,1}) > v_2(B^1_{t,2}),\\
    &v_1(B^2_{t,1}) < v_1(B^2_{t,2}), &v_2(B^2_{t,1}) < v_1(B^2_{t,2}).
    \end{align*}

    In words, both agents prefer $B^1_{t,1}$ to $B^1_{t,2}$, and $B^2_{t,2}$ to $B^2_{t,1}$. This can be assumed due to the fact that the allocations cancelling out allows us to know that exactly one agent feels envy in each allocation (if they both felt envy in some allocation then the other allocation would have to be EF in order to cancel out), and the same agent cannot feel envy in both allocations (or else their envy would clearly not cancel out). For simplicity we assume that in this case, for each $i \in [2]$, agent $i$ feels envy in allocation $B^{3-i}_{t}$, and allocation $B^i_t$ can be thought of as agent $i$'s ``preferred'' allocation. We can further conclude all the above inequalities must be strict, since if any agent was indifferent between the bundles of one allocation, they could not feel any envy in the other allocation, making one of the two allocations EF.

    \Cref{alg:envy-balancing} functions by examining each day in order, and picking one allocation from each day's pair. When the algorithm selects an allocation for some day, it puts it into one of two sets. $F$ is the ``Final'' set. If an allocation is put into $F$, that means that it will be in the final allocation returned by the algorithm. $S$ is the ``Swap'' set. If an allocation is put into $S$, that means that it may be changed at some point in the future. Specifically, the algorithm may perform a SWAP on $S$. This means that for every day $t \in [k]$, if an allocation from day $t$ is in $S$, that allocation will be replaced with the other allocation from day $t$ that is not in $S$. The final allocation returned by the algorithm is induced by all the per day allocations currently in $F \cup S$ after the algorithm's final iteration. We will refer to $S_t$ and $F_t$ as the contents of the sets $S$ and $F$ directly after iteration $t$ of the algorithm has completed, and will refer to SWAP$(S)$ as the contents of $S$ if a SWAP were performed on it. With slight abuse of notation, for any set $T$ containing allocations over some days, we will refer to $A_{T}$ as the allocation induced by combining all the per day allocations in $T$.

    For a day $t$ where one of $(B^1_t,B^2_t)$ is an EF allocation, the algorithm will add that allocation to $F$. Otherwise, the algorithm will check which agent is currently feeling envy in $A_S$, and will add that agent's preferred allocation to $S$ (if neither agent feels envy, then the algorithm defaults to adding $B^1_t$). The key part of the algorithm is what happens when either $A_S$ becomes EF, or if neither agent feels surplus utility in $A_S$. When $A_S$ is EF, then the algorithm moves the contents of $S$ into $F$, locking in those allocations as final. If neither agent feels surplus utility in $A_S$, then the algorithm performs a SWAP on $S$, and then moves the newly swapped contents of $S$ into $F$. As will be seen in the analysis below, this process ensures $A_F$ always remains EF, and $A_S$ always remains EF1. Note that this process also ensures that if one of the allocations for some day $t$ is EF, no allocation for that day will ever appear in $S$. Therefore, all allocations in $S$ will conform to the properties we were able to assume above.

    Clearly \Cref{alg:envy-balancing} will produce an allocation $A$ such that $A_{M_t} \in \set{B^1_t,B^2_t}$ for each day $t \in [k]$, making it EF1 per day. We will show that it also produces an allocation that is EF1 up to each day.

    To do this, we will first note that for both agents $i \in [2]$, for any possible set $S$ during the runtime of the algorithm, if an agent does not feel surplus utility in the allocation $A_S$, then they will not feel envy in the allocation $A_{\text{SWAP}(S)}$. We can show this formally by showing that the two allocations $A_S$ and $A_{\text{SWAP}(S)}$ cancel out. Let $D \subseteq [k]$ be the set of days such that an allocation for day $t$ appears in $S$. From the fact that we know the pair of allocations $(B^i_t,B^{3-i}_t)$ on each day cancels out, we have:

    \[\sum_{t \in D}{(v_i(B^i_{t,i}) + v_i(B^{3-i}_{t,i}))} \geq \sum_{t \in D}{(v_i(B^i_{t,3-i}) + v_i(B^{3-i}_{t,3-i}))}\]

    Then note that we have $v_i(A_{S,i}) + v_i(A_{\text{SWAP}(S),i}) = \sum_{t \in D}{(v_i(B^i_{t,i}) + v_i(B^{3-i}_{t,i}))}$, and $v_i(A_{S,3-i}) + v_i(A_{\text{SWAP}(S),3-i}) = \sum_{t \in D}{(v_i(B^i_{t,3-i}) + v_i(B^{3-i}_{t,3-i}))}$. This is because for each $t \in D$, both $B^i_t$ and $B^{3-i}_t$ are contained in exactly one of $S$ or SWAP$(S)$. This directly implies that $A_S$ and $A_{\text{SWAP(S)}}$ cancel out.

    Along with the above facts, proving the following inductive hypothesis will be sufficient to show that the allocation returned by \Cref{alg:envy-balancing} will be EF1 up to each day: 

    For all $t \in [k]$, if the following conditions hold after the $(t-1)$th iteration of the algorithm, they will hold after the $t$th iteration.
    \begin{itemize}
        \item $A_{F_t}$ is an EF allocation.
        \item $A_{S_t}$ is an EF1 allocation.
        \item $A_{\text{SWAP}(S_t)}$ is an EF1 allocation.
    \end{itemize}

    This inductive statement being true implies that the final allocation will be EF1 up to each day due to the fact that in the final allocation $A$ outputted by the algorithm, for any $t \in [k]$, it must be true that the allocation $A_{\bar{M}_t} \in \set{A_{F_t \cup S_t}, A_{F_t \cup \text{SWAP}(S_t)}}$. This is because the only way the algorithm ever interacts with allocations in $S$ from previous iterations is by performing a SWAP on $S$, or moving the entire contents of $S$ into $F$. Due to the fact that we know $A_{F_t}$ will be EF and $A_{S_t}$ and $A_{\text{SWAP}(S_t)}$ must be EF1, we can conclude that $A_{\bar{M}_t}$ must be EF1 as well.

    We will show that the inductive statement holds by analyzing each possible state the algorithm can be in after some iteration $t$.

    First, we will show that this holds for some obvious cases. 
    \begin{itemize}
        \item \paragraph{Day $t$ has an EF Allocation} When the pair $(B^1_t,B^2_t)$ for some day $t$ contains an EF allocation, then the algorithm simply adds that allocation to $F$. This clearly maintains the envy-freeness of $F_t$, and the contents of $S_t$ will be the same as $S_{t-1}$.
        \item \paragraph{During iteration $t$, some allocation $B^i_t$ is added to $S$ that causes $A_{S_{t-1} \cup \set{B^i_t}}$ to be EF} Directly after $B^i_t$ has been added to $S$, $A_{S}$ will be an EF allocation, and the algorithm will move the entire current contents of $S$ to $F$. $A_{F_t}$ will remain EF since the algorithm is adding an EF allocation to it. $A_{S_t}$ and $A_{\text{SWAP}(S_t)}$ will trivially meet their conditions since $S_t$ will be empty.
        \item \paragraph{During iteration $t$, an allocation $B^i_t$ is added to $S$ that causes neither agent to feel surplus utility in $A_{S_{t-1} \cup \set{B^i_t}}$} Similar to above, directly after $B^i_t$ has been added to $S$, both agents will either feel envy in $A_S$, or will be indifferent between the two bundles in $A_S$. In this case, the algorithm will perform a SWAP on $S$. The allocation induced by this newly swapped $S$ will be EF. The algorithm will then move the contents of $S$ to $F$. $A_{F_t}$, $A_{S_t}$, and $A_{\text{SWAP}(S_t)}$ will meet the required conditions for the same reasons as in the case above.
        \item \paragraph{$S_{t-1} = \emptyset$} Finally, in the case where $S$ is empty at the beginning of iteration $t$, and neither of the two allocation for day $t$ are EF, then the algorithm will add the allocation $B^1_t$ to $S$. $A_{F_t}$ will be EF since it was not altered, $A_{S_t}$ will be EF1 since $B^1_t$ is EF1, and $\text{SWAP}(A_{S_t})$ will be EF1 since $B^2_t$ is EF1.
    \end{itemize}

    From this, whenever the algorithm executes iteration $t$ and was not in one of the above cases, we can conclude the following:
    \begin{itemize}
        \item Neither of $B^1_t$ or $B^2_t$ are EF.
        \item $S_{t-1}$ was not empty. Since we know that during iteration $t-1$, if an allocation is added to $S$ that makes neither agent feel envy or neither agent feel surplus utility, then the contents of $S$ would be moved to $F$, causing $S_{t-1}$ to be empty. This allows us to conclude that exactly one agent must feel envy in $A_{S_{t-1}}$, and the other must feel surplus utility (notably, neither agent can be indifferent between the bundles).
        \item $S_t$ will not be empty, due to the fact that $S$ only becomes empty when either $A_S$ is EF, or neither agent feels surplus utility in $A_S$. Similarly to the above point, this allows us to conclude that exactly one agent must feel envy in $A_{S_{t}}$, while the other feels surplus utility. It also allows us to conclude that a SWAP was not performed in iteration $t$, as a SWAP only occurs in the case where neither agent feels surplus utility in $A_S$.
    \end{itemize}

    We can show that in this case as well, the inductive step holds.

    In this case, there will be one agent $i \in [2]$ who feels envy in $A_{S_{t-1}}$, while the other agent feels surplus utility. During iteration $t$, the algorithm will select agent $i$'s preferred allocation $B^i_t$, and add it to $S$. By our hypothesis, we have that $A_{S_{t-1}}$ is EF1, so there must be some good $g \in A_{S_{t-1},3-i}$ that can be taken away to eliminate all agent $i$'s envy in that allocation. Because we add in agent $i$'s preferred allocation from day $t$, we know that $v_i(A_{S_t,i}) - v_i(A_{S_t,3-i}) > v_1(A_{S_{t-1},i}) - v_1(A_{S_{t-1},3-i})$, meaning that $A_{S_t}$ must still be EF1 with respect to agent $i$, as we can still remove $g$ from $A_{S_t,3-i}$ to eliminate all envy. We can also show that $A_{S_t}$ will be EF1 with respect to agent $3-i$. $B^{i}_t$ is agent $3-i$'s unpreferred allocation, however, we know that it is EF1. Therefore, there exists some $g \in B^i_{t,i}$ such that $v_{3-i}(B^i_{t,3-i}) \geq v_{3-i}(B^i_{t,i}) - v_{3-i}(g)$. Combining this with the fact that $v_{3-i}(A_{S_{t-1},3-i}) > v_{3-i}(A_{S_{t-1},i})$, we get that $v_{3-i}(A_{S_{t-1},3-i}) + v_{3-i}(B^i_{t,3-i}) > v_{3-i}(A_{S_{t-1},i}) + v_{3-i}(B^i_{t,i}) - v_{3-i}(g)$. Since we know that $g \in A_{S_{t},i}$, this gives us EF1 as desired.

    The proof that $A_{\text{SWAP}(S_t)}$ is EF1 can be done similarly. Notice that SWAP$(S_t)$ will be equal to $\text{SWAP}(S_{t-1}) \cup \{B^{3-i}_t\}$. We know $A_{\text{SWAP}(S_{t-1})}$ will be EF1, and in $A_{\text{SWAP}(S_{t-1})}$, we know that agent $i$ will feel surplus utility. Agent $3-i$ may feel envy in $A_{\text{SWAP}(S_{t-1})}$, but since $B^{3-i}_t$ is their preferred allocation, they cannot be the reason why $A_{\text{SWAP}(S_t)}$ is not EF1. Agent $i$ also cannot be the reason why $A_{\text{SWAP}(S_t)}$ is not EF1, since there must be some good $g \in B^{3-i}_t$ that will eliminate all of agent $i$'s envy over $A_{\text{SWAP}(S_t)}$ when removed.

    Finally, note that $A_{F_t}$ will be EF since we have $A_{F_t} = A_{F_{t-1}}$. This shows the the inductive step will hold in this final case.
    
    As a base case, we look at the beginning of the algorithm, where both $S$ and $F$ will be empty, and thus all the conditions in the inductive statement will trivially be true.
\end{proof}

Note that the lemma achieves EF1 up to each day while not only retaining the EF1 per day property of the input allocations, but in fact by using exactly one of the two input allocations on each day. Thus, if the per day allocations given as input satisfy properties stronger than EF1, those properties are also retained per day; this is important as we will use this lemma to derive such stronger per day guarantees.

We are now ready to show that the pair of allocations required by the envy-balancing lemma --- both satisfying EF1 and cancelling each other out --- exists and can be computed in polynomial time. In fact, we find a single partition $(B_{t,1},B_{t,2})$ of the goods in $M_t$ such that both $B_t = (B_{t,1},B_{t,2})$ and $B'_t = (B_{t,2},B_{t,1})$ satisfy the stronger property of SD-EF1. 

\begin{lemma}\label{lem:2-opposite-sdef1}
    Given the preferences of two agents over a set of goods $M_t$, one can efficiently compute a partition $(B_{t,1},B_{t,2})$ of $M_t$ such that both $B_t = (B_{t,1},B_{t,2})$ and $B'_t = (B_{t,2},B_{t,1})$ are SD-EF1 allocations. 
\end{lemma}
\begin{proof}
    Draw a graph $G = (M_t,E)$ with the goods in $M_t$ as the nodes. For each agent $i \in [2]$ and $r \in [\floor{\frac{|M_t|}{2}}]$, draw an edge between the two goods in $T_i(M_t,2r) \setminus T_i(M_t,2(r-1))$. The edges added for each agent form a matching, so this graph is a union of two matchings, and hence, a bipartite graph. Thus, it admits a 2-coloring $c : M_t \to \set{1,2}$, which can be computed efficiently. Define $B_t = (B_{t,1},B_{t,2})$, where $B_{t,i} = \set{g \in M_t : c(g)=i}$ for each $i \in [2]$. 
    
    Note that due to the way we added the edges, each agent $i \in [2]$ receives exactly $r$ of her $2r$ most favorite goods, for each $r \in [\floor{\frac{|M_t|}{2}}]$. This meets the sufficiency condition from \Cref{prop:sd-ef1}, implying that $B_t$ is SD-EF1. 

    It is easy to see that the same reasoning also shows $B'_t = (B_{t,2},B_{t,1})$ is also SD-EF1.
\end{proof}

For readers familiar with the bihierarchy matrix decomposition theorem of \citet{BCKM13}, it is worth remarking that \Cref{lem:2-opposite-sdef1} can also be derived as a corollary. Specifically, we can define a binary variable $x_g \in \set{0,1}$ to indicate whether good $g$ should be allocated to agent $1$ (with $1-x_g$ denoting whether it should be allocated to agent $2$, and write the set of constraints: 
\begin{align*}
&\forall r \in \left[\floor{\frac{|M_t|}{2}}\right] : \textstyle\sum_{g \in T_1(M_t,2r) \setminus T_1(M_t,2(r-1))} x_g = 1,\\
&\forall r \in \left[\floor{\frac{|M_t|}{2}}\right] : \textstyle\sum_{g \in T_2(M_t,2r) \setminus T_2(M_t,2(r-1))} (1-x_g) = 1.
\end{align*}
It is easy to notice that this constraint set forms a ``bihierarchy'', and since it admits a fractional solution ($x_g = 1/2$ for all $g$), the result of \citet{BCKM13} implies the existence and polynomial-time compatibility of an integral allocation satisfying them, which is what we need. However, we provide a more direct proof for our specific constraint set because it is simpler to understand and leads to a faster algorithm.

We can plug in these allocations into the envy-balancing lemma (\Cref{lem:envy-balancing}) to get our desired main result. 

\begin{theorem}\label{thm:2-sdef1per-ef1each}
    For temporal fair division with $n=2$ agents, an allocation that is SD-EF1 per day and EF1 up to each day exists and can be computed in polynomial time. 
\end{theorem}
\begin{proof}
     Consider the allocations $B_t$ and $B'_t$ generated by \Cref{lem:2-opposite-sdef1}. Each is SD-EF1 and because they use the same partition of $M_t$ into bundles but do the opposite assignments, they trivially cancel out. Thus, due to \Cref{lem:envy-balancing}, these can be combined to compute an allocation that is SD-EF1 per day and EF1 up to each day. 
\end{proof}

There are several other desirable fairness desiderata which, just like SD-EF1, are strengthenings of EF1. For any such desideratum, if we can prove that for any set of goods $S$, there always exist $2$ allocations that achieve that desideratum and cancel out, then that directly implies that is it possible to achieve that desideratum per day alongside EF1 up to each day for instances with two agents. In \Cref{app:sec-4}, we provide a discussion of some of these desiderata (namely EFX and EF1+PO), and for each, prove that the desired $2$ allocations that cancel out always exist.

\subsection{Impossibilities}\label{sec:two-negative}

We have shown that when there are only two agents, the strong guarantee of EF1 up to each day can be obtained along with the guarantee of SD-EF1 per day. However, one wonders if even stronger guarantees are possible, such as strengthening EF1 up to each day to SD-EF1 up to each day. We find that not only is this strengthening impossible, but even if we relax SD-EF1 up to each day to SD-EF1 overall, it is impossible to achieve alongside EF1 per day. Together, these two impossibility results prove that our guarantee of SD-EF1 per day and EF1 up to each day from \Cref{thm:2-sdef1per-ef1each} is the strongest possible in the hierarchy shown in \Cref{fig:hierarchy}. 

\begin{theorem}\label{thm:2-sdef1upto-NO}
    For temporal fair division with $n=2$ agents, SD-EF1 up to each day cannot be guaranteed.
\end{theorem}
\begin{proof}
    Consider the following instance in which four goods arrive over three days: $M_1=\set{g_1,g_4}$, $M_2=\set{g_3}$, and $M_3=\set{g_2}$. Two agents have identical valuations given by $v(g_1)=4$, $v(g_2)=3$, $v(g_3)=2$, and $v(g_4)=1$ (since we seek SD-EF1, only the fact that the agents strictly prefer $g_1 \succ g_2 \succ g_3 \succ g_4$ matters). 

    For an allocation $A$ to be SD-EF1 up to each day, $A_{\bar{M}_t}$ must be SD-EF1 for each $t \in [3]$, where $\bar{M}_1 = \set{g_1,g_4}, \bar{M}_2 = \set{g_1,g_3,g_4}$, and  $\bar{M}_3 = M = \set{g_1,g_2,g_3,g_4}$.

    By the necessity condition of \Cref{prop:sd-ef1}, we can make the following claims:
    \begin{itemize}
        \item For $A_{\bar{M}_1}$ to be SD-EF1, $g_1$ and $g_4$ must be given to different agents.
        \item For $A_{\bar{M}_2}$ to be SD-EF1, $g_1$ and $g_3$ must be given to different agents. 
        \item For $A_{\bar{M}_3}$ to be SD-EF1, $g_1$ and $g_2$ must be given to different agents. Since each agent must also get $2$ goods from $T(M,4) = M$, it follows that $g_3$ and $g_4$ must also be given to different agents.
    \end{itemize}
    
    We now have the requirements that $g_1$, $g_3$, and $g_4$ must all be pairwise given to different agents, which is impossible since there are only two agents. Hence, in this instance, there is no allocation that is SD-EF1 up to each day.
\end{proof}

\begin{theorem}\label{thm:2-ef1per-sdef1over-NO}
    For temporal fair division with $n=2$ agents, EF1 per day and SD-EF1 overall cannot be guaranteed simultaneously. 
\end{theorem}
\begin{proof}
    Consider the following instance, in which eight goods arrive over three days: $M_1=\{g_1,g_5,g_7\}$, $M_2=\{g_2,g_4,g_6\}$, and $M_3=\{g_3,g_8\}$. The valuations are as follows:
    \[
    \begin{array}{@{}llllllll@{}}
    v_1(g_1)=8 & v_1(g_2)=7 & v_1(g_3)=6 & v_1(g_4)=5 & v_1(g_5)=4 & v_1(g_6)=3 & v_1(g_7)=2 & v_1(g_8)=1\\
    v_2(g_2)=8 & v_2(g_3)=7 & v_2(g_1)=6 & v_2(g_4)=5 & v_2(g_6)=4 & v_2(g_8)=3 & v_2(g_5)=2 & v_2(g_7)=1    
    \end{array}
    \]

    Notice that although each agent has a different strict ordering over $M$, their orderings restricted to the goods in any day $M_t$ are identical. Since there are only $2$ or $3$ goods given on any day and each agent has strict preferences over them, for allocation $A_{M_t}$ to be EF1, the same agent cannot receive both the items from $T(M_t,2)$ (otherwise, the other agent, who receives at most one good from $M_t$ that they value strictly less than each good in $T(M_t,2)$, would envy them even after the removal of one of the goods).

    Thus, EF1 per day requires that among the pairs of goods $(g_1,g_5)$, $(g_2,g_4)$, and $(g_3,g_8)$, the two goods in each pair go to different agents.

    Now consider what is required in order for the allocation $A_M$ to be SD-EF1 overall while adhering to these constraints. Since both agents have strict preferences over the goods, by the necessity condition of \Cref{prop:sd-ef1}, we can say that for both $i \in \{1,2\}$, $|A_{M,i} \cap T_i(M,r)| \ge \floor{r/n}$ must be true for all $r \in [|M|]$. 
    
    This means that each agent must receive at least one of their top $2$ goods. Specifically, agent $1$ must receive one of $T_1(M,2) = \{g_1,g_2\}$, and agent $2$ must receive at least one of $T_2(M,2) = \{g_2,g_3\}$. 
    
    It also means that each agent must receive at least $2$ of their top $4$ goods. Observe that $T_i(M,4) = \{g_1,g_2,g_3,g_4\}$ for both agents, so, in order to satisfy this, each agent must receive \emph{exactly} $2$ goods from $\{g_1,g_2,g_3,g_4\}$. 
    
    Notice that under these restrictions, in any allocation where agent $1$ receives $g_2$, they cannot receive $g_3$ (or they would have both of agent $2$'s top $2$ goods) or $g_4$ (due to the per day constraints). Since the allocation over the top $4$ goods must be balanced (each agent receiving exactly two of them), the only possible allocation $A_{T_1(M,4)}$ in this scenario would be $(\{g_1,g_2\},\{g_3,g_4\})$. Using the same logic, when agent $2$ receives $g_2$, the only possible allocation is $(\{g_1,g_4\},\{g_2,g_3\})$. Since one of the agents must be given $g_2$, these are the only two ways that the top $4$ goods can be allocated. Note that in both of these allocations, agent 1 gets $g_1$ and agent 2 gets $g_3$.

    Finally, consider how the remaining goods $\{g_5,g_6,g_7,g_8\}$ must be allocated to guarantee SD-EF1 overall. Each agent must receive at least $3$ of their top $6$ goods, and since we know that each agent has exactly $2$ of their top $4$ goods, that means that agent $1$ must receive at least one of $\{g_5,g_6\}$, and agent $2$ must receive at least one of $\{g_6,g_8\}$. Since we know that agent 1 must be allocated $g_1$, the per day constraints say they cannot receive $g_5$, so they must receive $g_6$. Similarly, agent $2$ is known to have $g_3$, so they cannot receive $g_8$, which means they must also receive $g_6$. Clearly, both agents cannot receive $g_6$, meaning in this instance, there is no allocation that is EF1 per day and SD-EF1 overall. 
\end{proof}

\subsection{Two Agents and Further Restrictions}

As the final part of this section, we note that when further restrictions are placed on two agent instances, we receive even stronger positive results. Particularly, in \Cref{app:sec-4}, we prove and provide a discussion of the following theorem.

\begin{theorem}\label{2-iddays-SDEF1}
    For temporal fair division with $n=2$ agents and identical days, an allocation that is SD-EF1 per day, SD-EF1 up to each day, and SD-EF up to each even day exists and can be computed in polynomial time. 
\end{theorem}

%%%%%%%%%%%%%%%%%%%%%%%%%%%%%%%%%%%%%%%%%%%%%%%%%%%%%%%%%%
%%%%%%%%%%%%%%%%%%%%%%%%%%%%%%%%%%%%%%%%%%%%%%%%%%%%%%%%%%
%%%%%%%%%%%%%%%%%%%%%%%%%%%%%%%%%%%%%%%%%%%%%%%%%%%%%%%%%%

\section{Identical Orderings}\label{sec:identical}

We next look at instances where agents have identical orderings over all goods in $M$. Not only are results in this scenario practical useful, as there many real life scenarios where participants agree on the the ordinal ranking of goods, but results under identical orderings are also very technically useful. As can be seen from our main result in \Cref{sec:general}, reducing a general setting to one where agents have similar orderings over the goods can lead to fairness guarantees in the original setting. We will show in future sections that the possibility results we develop here can be applied as black-boxes to achieve strong results in more scenarios where agents have heterogeneous orderings.

\paragraph{Possibilities.}
Our main result for the case of identical orderings is the following theorem:

\begin{theorem}\label{thm:idpref-sdef1per-sdef1over}
    For temporal fair division with identical preferences, an allocation that is SD-EF1 per day and SD-EF1 overall exists and can be computed in polynomial time.
\end{theorem}
\begin{proof}
    Start by constructing two set families, $P_1$ and $P_2$. 

    \begin{align*}
        P_1 &= \{T(M_t,nr) \setminus T(M_t,n(r-1)) : r \in \left[\ceil{\frac{|M_t|}{n}}\right], t \in [k]\},\\
        P_2 &= \{T(M,nr) \setminus T(M,n(r-1)) : r \in \left[\ceil{\frac{|M|}{n}}\right]\}.
    \end{align*}

    In words, $P_2$ splits the entire set of goods $M$ into the agents' most preferred $n$ goods, their next most preferred $n$ goods, etc. $P_1$ does a similar partitioning, but splits the goods from each day separately.
    
    Because preferences orderings of all agents are identical, we can use the sufficiency condition of \Cref{prop:sd-ef1}, which states that $\card{A_i \cap T(M,r)} \in \set{\floor{r/n},\ceil{r/n}}$ for all $i \in N$ and $r \in [|M|]$ implies SD-EF1 overall. Clearly, if each agent receives at most $1$ good from each set in $P_2$, this will be true. Similarly, if each agent receives at most $1$ good from each set in $P_1$, then for all $i \in N$, $t \in [k]$, and $r \in [|M_t|]$, we would have $\card{A_i \cap T(M_t,r)} \in \set{\floor{r/n},\ceil{r/n}}$, implying SD-EF1 per day. Therefore, any allocation which gives each agent at most $1$ good from each of the sets in $P_1 \cup P_2$ will be SD-EF1 per day and SD-EF1 overall. Our goal is to find an allocation that meets these constraints.

    We will first start by adjusting the structure of the set families slightly to make the problem easier to work with. Note that it must be the case that $|P_1| \geq |P_2|$. $P_2$ will contain $\ceil{m/n}$ sets, $\floor{m/n}$ sets of exactly $n$ goods, and if $m/n$ is not an integer, then $1$ additional set will be included containing the remaining goods. $P_1$ will contain at least $\ceil{m/n}$ sets, since it will be a disjoint partition of all the goods in $m$, with each set having a maximum set size of $n$, but may contain up to $k$ sets with size less than $n$ due to the fact that it is partitioning each day individually.

    First, we will create $|P_1|-|P_2|$ empty sets and place them in $P_2$, making it so that $|P_1| = |P_2|$. Then, if $|P_1| > \frac{m}{n}$, we will create $n|P_1| - m$ dummy goods (goods with $0$ value to every agent), and place each dummy good into a set from $P_1$ and $P_2$ in such a way that all $|P_1|$ sets in both families contain exactly $n$ goods each. Note that the restriction that each agent must receive no more than $1$ good from each of these updated families creates is a stronger constraint than what is needed to guarantee SD-EF1. Any allocations that meets the constraints on the updated set families can be clearly shown to meet the original SD-EF1 constraints simply by removing all the dummy goods. After these updates, both $P_1$ and $P_2$ will contain the same number of sets, and the size of every set in both families will be exactly $n$.
    
    Next, construct a graph where each set from $P_1$ or $P_2$ is a vertex, and put an edge between two sets in $P_1$ and $P_2$ if they share a good (there can be multiple edges between vertices). Label each edge to correspond to the good it represents. Since both $P_1$ and $P_2$ are disjoint partitions over the set of goods $M + \text{``the dummy goods''}$, there will be no edge between two sets from the same family, making this a bipartite graph where each vertex has a degree of exactly $n$. Therefore, this bipartite graph can be deconstructed into $n$ perfect matchings. Let each of these matchings denote a bundle of the items corresponding to the edges in that matching. These $n$ bundles will form a disjoint partition over the full set of goods $M$, and each bundle will contain at most $1$ good from each set in $P_1 \cup P_2$.
    
    Since there are $n$ bundles, we can create an allocation by assigning each one arbitrarily to an agent. Any allocation formed this way will meet the constraints for both $P_1$ and $P_2$, and thus will be be SD-EF1 per day and SD-EF1 overall.
\end{proof}

To an initiated reader, the problem of finding an allocation that meets the above constraints may be immediately reminiscent of the bihierarchy framework of \citet{BCKM13}. We require that from each set of (at most) $n$ goods out of a family of sets, each agent receives (at most) one good. The sets produced by each desideratum are mutually non-overlapping, forming a ``hierarchy'', but the sets produced by one desideratum can be overlapping with those produced by the other, resulting in two different hierarchies. However, the problem is that with $n > 2$ agents, we have a third set of constraints: each good must be assigned to (exactly) one agent. This forms a third hierarchy (which cannot be assimilated into either of the two previous hierarchies), preventing one from applying the bihierarchy framework. 

\paragraph{Impossibilities.}

\begin{theorem}\label{thm:idpref-iddays-sdef1uptoeach-NO}
    For temporal fair division with identical days and identical preferences, SD-EF1 up to each day cannot be guaranteed.
\end{theorem}
\begin{proof}
    Consider the following instance in which twenty-four goods arrive over four identical days: $M_t=\set{g_{1,t},g_{2,t},g_{3,t},g_{4,t},g_{5,t},g_{6,t}}$ for all $t \in [4]$. Twelve agents have identical valuations given by $v(g_{1,t})=6$, $v(g_{2,t})=5$, $v(g_{3,t})=4$, $v(g_{4,t})=3$, $v(g_{5,t})=2$, $v(g_{6,t})=1$ (since we seek SD-EF1, it only matters that the agents strictly prefer $g_{1,t} \succ g_{2,t} \succ g_{3,t} \succ g_{4,t} \succ g_{5,t} \succ g_{6,t}$ for all $t \in [4]$, and are indifferent between any two goods $g_{l,t},g_{l,t'}$ for all $l \in [6], t,t' \in [4]$. 

    Consider how the requirement of SD-EF1 up to each day restricts how each good can be allocated:
    \begin{itemize}
        \item The bundle $\bar{M}_2$ will contain $12$ goods. By the necessity condition of \Cref{prop:sd-ef1}, we know for $A_{\bar{M}_2}$ to satisfy SD-EF1, it must be the case that each agent receives exactly $1$ good from $T(\bar{M}_2,12) = \bar{M}_2$.
        \item For $A_{\bar{M}_3}$ to be SD-EF1, it must be the case that each agent receives exactly $1$ of the goods in $T(\bar{M}_3,12)$. This is because $T(\bar{M}_3,12)$ will contain goods $g_{1,t},g_{2,t},g_{3,t},g_{4,t}$ for all $t \in [3]$. The next good in the agents' ordering will be a copy of $g_{5,t}$ for some $t$, so all the goods in $T(\bar{M}_3,12)$ are strictly preferred over the remaining goods in $\bar{M}_3 \setminus T(\bar{M}_3,12)$.
        \item For $A_{\bar{M}_4}$ to be SD-EF1, it must be the case that each agent gets exactly $1$ of the goods in $T(\bar{M}_4,12)$. This is because $T(\bar{M}_4,12)$ will contain goods $g_{1,t},g_{2,t},g_{3,t}$ for all $t \in [4]$. Similar to the previous case, the next good in the agents' ordering will be a copy of $g_{4,t}$ for some $t$, so it must be the case that all the goods in $T(\bar{M}_4,12)$ are strictly preferred over the remaining goods in $\bar{M}_4 \setminus T(\bar{M}_4,12)$. Since there are $24$ total goods in $\bar{M}_4$, we can also say that each agent must receive exactly $2$ goods from the set $T(\bar{M}_4,24) = \bar{M}_4$ . Since no agent can have more than $1$ good from $T(\bar{M}_4,12)$, the only way to satisfy this condition is to give each agent exactly $1$ good from $T(\bar{M}_4,24) \setminus T(\bar{M}_4,12)$.
    \end{itemize}

    \usetikzlibrary{arrows.meta} 
    \usetikzlibrary{matrix,shapes,arrows,fit}
    \begin{figure}[htb]
        \centering
        \begin{tikzpicture}%[thick]
          \matrix (M) [%
            matrix of nodes, column sep=3pt, row sep=3pt
          ]
          {%
            { }& $t=1$ & $t=2$ & $t=3$ & $t=4$\\[0.2cm]
             $g_1$ & $1$ & $7$ & $/$ & $/$ \\
             $g_2$ & $2$ & $8$ & $/$ & $/$ \\
             $g_3$ & $3$ & $9$ & $/$ & $/$  \\
             $g_4$ & $4$ & $10$ & $/$ & $/$ \\
             $g_5$ & $5$ & $11$ & $/$ & $/$ \\
             $g_6$ & $6$ & $12$ & $/$ & $/$ \\
           };
           \node[draw=blue,rounded corners = 1ex,fit=(M-2-3)(M-7-2),inner sep = 4pt] {};
           \node[draw=green!50!black,rounded corners = 1ex,fit=(M-5-2)(M-2-4),inner sep = 2pt] {};
           \node[draw=red,rounded corners = 1ex,fit=(M-4-2)(M-2-5),inner sep = 0pt] {};
           \node[draw=red,rounded corners = 1ex,fit=(M-7-2)(M-5-5),inner sep = 0pt] {};
        \end{tikzpicture}
        \caption{Restrictions placed on the allocation by SD-EF1 up to any day}
        \label{fig:12-6-6}
    \end{figure}

    \Cref{fig:12-6-6} makes these restrictions clear. The restrictions over $\bar{M}_2$ are shown in Blue, $\bar{M}_3$ in Green, and $\bar{M}_4$ in Red. Each box represents a group of $12$ goods that all must go different agents, and the entries in the first two table columns shows a hypothetical allocation of $\bar{M}_2$. Notice that the good $g_{4,3}$ is contained in a restriction from $\bar{M}_3$ and $\bar{M}_4$. Due to its $\bar{M}_3$ restriction, it cannot be given to an agent that has been assigned a good from $T(\bar{M}_2,8)$. Due to its $\bar{M}_4$ restriction, it cannot be given to an agent that been assigned a good from $T(\bar{M}_2,12)\setminus T(\bar{M}_2,6)$. However, these two groups together make up the entire set of goods $\bar{M}_2$, which contains $12$ goods, each assigned to one of the $12$ agents. So there is no agent we can assign $g_{4,3}$ to that will lead to the satisfaction of SD-EF1 for $A_{\bar{M}_2}$, $A_{\bar{M}_3}$ and $A_{\bar{M}_4}$.
\end{proof}

%%%%%%%%%%%%%%%%%%%%%%%%%%%%%%%%%%%%%%%%%%%%%%%%%%%%%%%%%%
%%%%%%%%%%%%%%%%%%%%%%%%%%%%%%%%%%%%%%%%%%%%%%%%%%%%%%%%%%
%%%%%%%%%%%%%%%%%%%%%%%%%%%%%%%%%%%%%%%%%%%%%%%%%%%%%%%%%%

\section{Identical Days}\label{sec:samedays}

In \Cref{sec:general}, we showed that in the general model, we can achieve SD-EF1 per day and PROP1 overall. In this section, we will show if we assume the additional restriction that the sets of goods which arrive on each day are identical, then a slightly stronger guarantee can be achieved overall.

\begin{theorem}\label{thm:iddays-SDPROP1-overall}
    For any temporal fair division instance with identical days, it is possible to find an allocation that is SD-EF1 per day and SD-PROP1 overall in polynomial time.
\end{theorem}

To achieve this guarantee, we use an algorithm that is almost identical to \Cref{alg:one}, the algorithm which was used to achieve SD-EF1 per day and PROP1 overall in the general case, but with one major change.

\Cref{alg:one} first finds an EF1 overall allocation in a reduced version of the problem where agents have identical orderings over all the goods on any fixed day, and uses that to construct an allocation in the original instance that is PROP1 overall. The key insight that can be leveraged to get stronger guarantees is that when we have identical days, we know that the output of the ``Identical Ordering Transformation'' from \Cref{alg:one} will result in an instance where all agents have identical orderings over the entire set of goods $M$, not just over the goods from each individual day. This stronger guarantee from the Identical Ordering Transformation allows us to use \Cref{thm:idpref-sdef1per-sdef1over} to find an SD-EF1 per day and SD-EF1 overall allocation over the reduced instance (rather than the algorithm of \citep{biswas2018fair} which only guarantees EF1 overall). We can then use this allocation as the basis for the picking order that \Cref{alg:one} uses to construct the final allocation over the original instance. In \Cref{lem:alg1-prop1}, we showed that if \Cref{alg:one} produced an allocation over the Identical Ordering Transformation instance that was EF1, then running the picking order procedure on that instance would generate a final allocation that is PROP1. We can also show that if the allocation over the reduced instance is SD-EF1, then the allocation produced by the picking order procedure will be SD-PROP1. Below, we provide the proofs of these claims.

We do not include a formal statement of this algorithm, as it is nearly identical to \Cref{alg:one}. One can imagine \Cref{alg:one} with line 11, where the ``BarmanBiswasCC'' subroutine is used, replaced with the subroutine introduced in \Cref{thm:idpref-sdef1per-sdef1over} running on the Identical Ordering Transformation instance $(N,M',v')$.

\begin{lemma}\label{lem:iddays-reduction-sameorder}
    For any instance with identical days, the identical ordering transformation described in \Cref{alg:one} will produce an instance where all agents have identical orderings over the full set of goods.
\end{lemma}
\begin{proof}
    Given an identical days instance over a set of goods $M$, let $M'$ represent the full set of goods in the identical ordering transformation instance, and for each agent $i$, let $v'_i$ be $i$'s valuations over $M'$. For all $t \in [k]$, let $M'_t = \set{g'_{t,1},\dots,g'_{t,|M'_t|}}$ such that each agent has the identical weak ordering $g'_{t,1} \succeq \dots \succeq g'_{t,|M'_t|}$.

    Formally, we say that since the original instance has identical days, then for every $t,t' \in [k]$, there will be a bijection $f_{t,t'}: M_t \rightarrow M_{t'}$ such that for each $g \in M_t$ and each $i \in N$, $v_i(g) = v_i(f_{t,t'}(g))$. From the identical orderings transformation, we know that for each $t \in [k]$ and $i \in N$, there is a bijection $o_{i,t}: M_t \rightarrow M'_t$ such that $v_i(g) = v'_i(o_{i,t}(g))$ for each $g \in M_t$. Using these bijections, it must be the case that $v_i(g) = v_i(f_{t,t'}(g)) = v'_i(o_{i,t}(g)) = v'_i(o_{i,t}(f_{t,t'}(g)))$. Therefore, for each agent $i$, and for any days $t,t'$, there must exist a bijection $f'_{i,t,t'}: M'_t \rightarrow M'_{t'}$ such that $v'_i(g) = v'_i(f_{i,t,t'}(g))$ for all $g \in M'_t$. Since all agents have identical orderings over each day, then without loss of generality, we can assume that this bijection maps $g_{t,r}$ to $g_{t',r}$ for each $r \in |[M'_t]|$. This means that $f'_{i,t,t'}$ is identical for all $i$, and therefore the identical orderings transformation instance must also have identical days. We will call the bijection defining this property $f'_{t,t'}$.

    Using the above reasoning, we can say that for every $g,g' \in M'_1$, if $g \succeq g'$, then $f'_{1,t}(g) \succeq f'_{1,t}(g')$ for each $t \in [k]$. Since $\set{f'_{1,t}(g'_{1,r}) : t \in [k], r \in [|M'_1|]} = \cup_{t \in [k]}{M'_t} = M'$, this, along with the fact that $g \succeq f'_{1,t}(g)$ and $f'_{1,t}(g) \succeq g$ for each day $t$, induces a total weak preference ordering over $M'$. This ordering holds for each agent since all agents have the same preference ordering over $M'_1$.
\end{proof}

\begin{lemma}\label{lem:iddays-alg-gets-sdprop1}
    When provided with an allocation over an Identical Ordering Transformation instance that is SD-EF1 overall, the picking order procedure from \Cref{alg:one} will produce an allocation that is SD-PROP1 overall.
\end{lemma}
\begin{proof}
    Let $A'$ be the SD-EF1 allocation over the identical ordering transformation instance, defined by the set of goods $M'$ and the valuation functions $v'_i$ for all $i \in N$, and let $A$ be the allocation produced by $A'$ and the picking order sequence over the original instance, defined by the set of goods $M$ and the valuation functions $v_i$ for each $i \in N$. 

    This proof involves reasoning with two separate classes of bijections. The first is $o_i: M \rightarrow M'$, which is the bijection formed for each $i \in N$ by creating the identical orderings transformation in \Cref{alg:one}. In these bijections, we have that $v_i(g) = v'_i(o_i(g))$ for each $i \in N$ and $g \in M$, which implies that $v_i(M) = v'_i(M')$. The second class of bijections are $z_i: A_i \cup \set{g_i} \rightarrow A'_i \cup \set{g'_i}$, where $g_i, g'_i$ are agent $i$'s favourite good that they are not allocated in the original instance and the identical ordering transformation instance respectively. In these bijections, we will have that $v_i(g) \ge v'_i(z_i(g))$ for all $i \in N$ and $g \in A_i \cup \set{g_i}$, implying that $v_i(A_i \cup \set{g_i}) \ge v'_i(A'_i \cup \set{g'_i})$. The existence of this class of bijections are guaranteed by \Cref{lem:prop1-helper} and \Cref{lem:alg1-prop1}.

    Since SD-EF1 implies EF1, and $A'$ will be SD-EF1 overall, it is easy to see that we can invoke the same logic from \Cref{lem:prop1-helper} and \Cref{lem:alg1-prop1} to show that $A$ is PROP1 overall. To show that the allocation is in fact SD-PROP1 overall, we can do a more detailed analysis on the proof technique from \Cref{lem:alg1-prop1}. Consider the set of goods $M^* = M \cup M'$ and each agent $i$'s ordering over $M^*$ that is induced by their valuation functions $v_i$ and $v'_i$. We note that when considering the set $M^*$, the bijection $z_i$ not only shows that $v_i(A_i \cup \set{g_i}) \geq v'_i(A_i' \cup \set{g'_i})$, but also allows us to conclude the stronger claim that $A_i \cup \set{g_i} \succeq_i^\sd A_i' \cup \set{g'_i}$. This is because for each good $g \in A_i \cup \set{g_i}$, we have that $v_i(g) \geq v'_i(z_i(g))$, and therefore also have $g \succeq_i z_i(g)$.

    Next, note that for each agent $i \in N$, and each $g \in M$, we must have that $H_i(M',o_i(g)) = \set{o_i(g') : g' \in H_i(M,g)}$, since clearly, $H_i(M',o_i(g))$ is the set of all goods $g' \in M'$ such that $g' \succeq_i o_i(g)$, and similarly, $H_i(M,g)$ is the set of goods $g' \in M$ such that $g' \succeq_i g$. From our definition of $o_i$, we have that for any $g' \in M$, $v_i(g') \geq v_i(g)$ if and only if $v'_i(o_i(g')) \geq v'_i(o_i(g))$.

    With this, we can conclude that $\card{H_i(M,g)} = \card{H_i(M',o_i(g))}$, and also that $\card{(A_i \cup \set{g_i}) \cap H_i(M,g)} \geq \card{(A'_i \cup \set{g'_i}) \cap H_i(M',o_i(g))}$ for all goods $g \in M$. The second statement follows from the fact that for any $g^* \in A_i \cup \set{g_i}$, if $z_i(g^*) \in H_i(M',o_i(g))$, then $g^* \in H_i(M,g)$ must be true as well, since we have that $v_i(g^*) \ge v'_i(z_i(g^*)) \ge v'_i(o_i(g)) = v_i(g)$.

    We know that $A'$ is SD-EF1, implying it is SD-PROP1.  Therefore, for all $g \in M$, we have that $\card{(A_i \cup \set{g_i}) \cap H_i(M,g)} \ge \card{(A'_i \cup \set{g'_i} \cap H_i(M',o_i(g))} \ge \ceil{\nicefrac{\card{H_i(M',o_i(g))}}{n}} = \ceil{\nicefrac{\card{H_i(M,g)}}{n}}$, allowing us to conclude that $A$ is SD-PROP1 overall.
\end{proof}

With these lemmas in mind, the proof of \Cref{thm:iddays-SDPROP1-overall} follows easily.

\begin{proof}[Proof of \Cref{thm:iddays-SDPROP1-overall}]
    From \Cref{lem:iddays-reduction-sameorder}, it is known that the identical days reduction will construct an instance where each agent has the same overall ordering for the full set of goods $M'$. This means that by the results of \Cref{thm:idpref-sdef1per-sdef1over}, an allocation $A'$ over $M'$ can be found that is SD-EF1 per day and SD-EF1 overall in polynomial time.

    Using identical logic from \Cref{lem:gen-alg1-sdef1}, constructing the final allocation $A$ by running the picking order construction procedure from \Cref{alg:one} on $A'$ will result in an allocation that is SD-EF1 per day, since $A'$ is SD-EF1 per day.

    Finally, from \Cref{lem:iddays-alg-gets-sdprop1}, we know that $A$ will also be SD-PROP1 overall, completing the proof.
\end{proof}

%%%%%%%%%%%%%%%%%%%%%%%%%%%%%%%%%%%%%%%%%%%%%%%%%%%%%%%%%%
%%%%%%%%%%%%%%%%%%%%%%%%%%%%%%%%%%%%%%%%%%%%%%%%%%%%%%%%%%
%%%%%%%%%%%%%%%%%%%%%%%%%%%%%%%%%%%%%%%%%%%%%%%%%%%%%%%%%%

\section{Beyond Temporal Fair Division}\label{sec:laminar}

Consider the following problem. A school district has received a new shipment of supplies and must distribute them between the schools under their jurisdiction. Each item is used for a different subject (there are microscopes for the biology lab, instruments for music class, easels and paint for art class, etc.), and each school has different preferences over the items. The district wants to find an allocation that is fair among the schools, but also wants to make sure that the allocations are fair with respect to each individual subject. For example, the district may be able to construct an EF1 allocation by giving all the biology supplies to one school, and all the art supplies to another school, but that would be extremely unfair to the individual departments within those schools. Interestingly, this problem is identical to finding an allocation over a temporal instance that is fair per day and fair overall, since both are in essence looking at a pairwise-disjoint partition of some set of items $M$, and finding an allocation that is fair with respect to $M$, and remains fair when only looking at any of the sets in the partition.

Up until this point, the best way to solve such a problem would be existing algorithms for constrained fair division, such as the cardinality constraint algorithm of \citep{biswas2018fair}. This algorithm would guarantee an EF1 overall allocation, while also guaranteeing that each school got a balanced number of items from each subject, but it gives no guarantee that the goods from each subject will be allocated fairly according to the valuation functions of each of the schools. While our solution, \Cref{alg:one}, only achieves PROP1 overall, it achieves the very strong SD-EF1 fairness guarantee for each subject, making it arguably a more desirable algorithm for use cases such as this. When there are $2$ agents, or when all agents' orderings are identical, then we get strictly better guarantees, since we can still achieve SD-EF1 per day (which implies balancedness), while also achieving EF1 overall.

While the ``per day'' and ``overall'' desiderata translate very well into this broader interpretation of the temporal fairness model, our other definition, ``up to each day'' does not. This is because ``up to each day'' implicitly assumes that there is some ordering over the sets of goods in the partition. One intuitive way to extend this concept is with laminar set families. We give a detailed discussion on laminar set families in \Cref{app:sec-7}, but the basic definition is that for any set of goods $M$, a collection $\mathcal{L}$ of subsets of $M$ is a laminar set family over $M$ if for every pair of subsets $S,T \in \mathcal{L}$, we have that either $S \cap T = \emptyset$, $S \subset T$, or $T \subset S$. Intuitively, they can be thought of as collection of subsets over $M$ that form a ``hierarchy structure''. In the school district example, the sets of items where fairness is required could be something like $\{\text{Biology},\text{Chemistry},\text{Drawing},\text{Music},\text{Science},\text{Arts}\}$, where Biology and Chemistry are subsets of Science and Drawing and Music are subsets of Arts. 

With this in mind, we can introduce a new concept of fairness, which generalizes ``up to each day'' and is more compatible with this abstract view of temporal fair division.

\begin{definition}[Laminar Fairness]
    For desideratum $X$, allocation $A$ is Laminar $X$ with respect to some laminar set family $\mathcal{L}$ if $A_{S}$ satisfies $X$ for all $S \in \mathcal{L}$.
\end{definition}

It can be seen that achieving both ``up to each day'' and ``per day'' fairness is a special case of Laminar Fairness. \Cref{fig:uptoeachdaylaminar} shows the laminar family induced by these constraints.

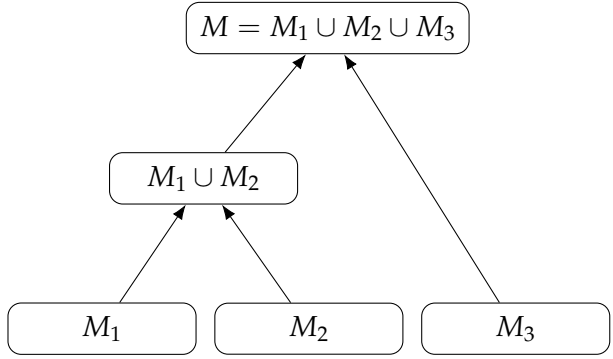
\begin{figure}[htb]
    \centering
        \begin{tikzpicture}
        \tikzset{
            mynode/.style={
                draw,
                rectangle,
                rounded corners=5pt,
                inner sep=5pt,
                minimum width=2.5cm
            },
            myarrow/.style={
                ->,
                -{Latex[length=2mm]}
            }
        }
        % Right side
        \node[mynode] (Y1) at (0,0) {$M_1$};
        \node[mynode] (Y2) at (5.5,0) {$M_3$};
        \node[mynode] (Y3) at (2.75,0) {$M_2$};
        \node[mynode] (Y4) at (1.35,2) {$M_1 \cup M_2$};
        \node[mynode] (Y5) at (3,4) {$M = M_1 \cup M_2 \cup M_3$};
        \draw[myarrow] (Y1) -- (Y4);
        \draw[myarrow] (Y3) -- (Y4);
        \draw[myarrow] (Y4) -- (Y5);
        \draw[myarrow] (Y2) -- (Y5);
        
        \end{tikzpicture}
    \caption{Representation of the ``up to each day'' and ``per day'' constraints for three days as a laminar set family. An allocation that is EF1 up to each day and EF1 per day would be an allocation that is EF1 with respect to every set in this family. To just represent the ``up to each day'' constraints, only the sets $M_1$, $M_1 \cup M_2$ and $M_1 \cup M_2 \cup M_3$ would be required.}
    \label{fig:uptoeachdaylaminar}
\end{figure}

We present an upgrade to our two agent Envy-Balancing algorithm from \Cref{lem:envy-balancing} that allows us to find EF1 allocations with respect to any laminar set family, a strictly stronger guarantee than simply finding allocations that are EF1 up to each day.

\begin{restatable}{theorem}{laminarenvybalancing}\label{lem:laminar-envy-balancing}
    Given $2$ agents, a set of goods $M$, and a laminar set family $\mathcal{L}$ over $M$, it is possible to find an allocation to those agents that is Laminar EF1 with respect to $\mathcal{L}$.
\end{restatable}

We give the proof of this lemma in \Cref{app:sec-7}.

Unfortunately, for the general case, laminar fairness will not always be possible. Since the ``up to each day'' constraint can be modeled as a laminar set family, all of the impossibility results for ``up to each day'' also hold for laminar fairness.

%%%%%%%%%%%%%%%%%%%%%%%%%%%%%%%%%%%%%%%%%%%%%%%%%%%%%%%%%%
%%%%%%%%%%%%%%%%%%%%%%%%%%%%%%%%%%%%%%%%%%%%%%%%%%%%%%%%%%
%%%%%%%%%%%%%%%%%%%%%%%%%%%%%%%%%%%%%%%%%%%%%%%%%%%%%%%%%%

\section{Discussion}\label{sec:discussion}

In this work, we are able to find possibility and impossibility results that give a picture of what can be achieved in the temporal fair division model. This picture is quite clear when focusing on special cases such as two agents or identical orderings. However, we still leave many questions open for future work, especially when looking at the most general instances of the model. The most tantalizing question left open in the general model is the following. 

\begin{open}
    In temporal fair division, does an allocation that is SD-EF1 (or EF1) per-day and EF1 overall always exist?
\end{open}

For the special case of identical days, studied by \citet{igarashi2024repeated}, several stronger guarantees are also open:

\begin{open}
    In temporal fair division with identical days, is it possible to achieve SD-EF1 per day and SD-EF1 overall?
\end{open}

In addition, while we obtain many results for the ``overall'' and ``per day'' temporal fairness properties, the existence of the ``up to each day'' property in many settings is still a mystery.

\begin{open}
    In temporal fair division with identical days or identical orderings, is it possible to achieve EF1 up to each day?
\end{open}

It would also be an interesting further direction to take a more abstract view of the temporal fair division model. We briefly explored this in \Cref{sec:laminar}, introducing a generalized model of temporal fair division called laminar fairness, and showing that it is obtainable for two agents. However, this still leaves several open questions regarding when laminar fairness is guaranteed to exist.

\begin{open}
    Do Laminar EF1 allocations exist when all agent have identical preferences? Do allocations that satisfy some weakening of Laminar EF1 (such as Laminar PROP1) exist in more general settings?
\end{open}

Finally, it would be beneficial to study any connections and implications of our model to the broader literature of temporal fairness beyond fair division, such as the general temporal fairness model of \citet{AKCM24}.

\section*{Acknowledgements}
This research was partially supported by an NSERC Discovery grant.

\bibliography{abb, nisarg, ultimate}

\newpage
\appendix
\section*{\centering Appendix}

%%%%%%%%%%%%%%%%%%%%%%%%%%%%%%%%%%%%%%%%%%%%%%%%%%%%%%%%%%
%%%%%%%%%%%%%%%%%%%%%%%%%%%%%%%%%%%%%%%%%%%%%%%%%%%%%%%%%%
%%%%%%%%%%%%%%%%%%%%%%%%%%%%%%%%%%%%%%%%%%%%%%%%%%%%%%%%%%

\section{Omitted from \Cref{sec:prelim}}\label{app:sec-2}

\subsection{Properties of SD-EF1 Allocations}
    We will begin by noting some useful properties of SD-EF1 allocations. Specifically, in \Cref{sec:prelim}, we establish the function $T_i(S,r)$, which returns agent $i$'s top $r$ ordered goods from a set $S$, breaking ties according to some arbitrary rule consistent across all agents. In contrast, the definitions for SD-EF1 and SD-PROP1 are based around sets of the form $H_i(S,g) = \set{g' \in S : g' \succeq_i g}$, which returns all goods from $S$ that are weakly preferred to $g$. As can be seen in many of this paper's proofs, it is often very useful to be able to look at a set of exactly size $r$ of some agent's top goods, which the sets $H_i(S,g)$ do not allow for. Below, we will show relations between the $T_i(S,r)$ function, and the SD-EF1 definition, that allows us to often use it without loss of generality when proving statements about SD-EF1 and SD-PROP1, vastly simplifying many of our proofs.

    \begin{observation}\label{obs:SDEF1}
        For any agent $i \in N$ and good $g \in S$, if $\card{H_i(S,g)} = r$, then $T_i(S,r) = H_i(S,g)$, regardless of the arbitrary tie-breaking order dictated by $T_i$.
    \end{observation}
    \begin{proof}
        This follows from the fact that agent $i$'s ordering $\succeq_i$ over the goods in $S$ will be transitive. $\card{H_i(S,g)} = r$ means that there are \emph{exactly} $r$ goods that $i$ weakly prefers to $g$. It must be the case that for any goods $g^+ \in H_i(S,g)$, $g^- \in S \setminus H_i(S,g)$, we must have that $g^+ \succ_i g^-$. If there were some goods $g^+ \in H_i(S,g)$, $g^- \in S \setminus H_i(S,g)$ such that $g^- \succeq_i g^+$, then by the transitivity of $i$'s ordering, we know that $g^- \succeq_i g^+ \succeq_i g$ is true, contradicting that fact that $g^- \in  S \setminus H_i(S,g)$.
        
        Therefore, when $\card{H_i(S,g)} = r$, we know that there are $r$ goods in $S$ that agent $i$ strictly prefers to all other goods in $S$. It is clear that $T_i(S,r)$ will return exactly those goods, and will not need to use its tie-breaking order in this case.
    \end{proof}

    With this observation in mind, we can now list several necessary and sufficient conditions for an allocation to be SD-EF1, which relate it directly to the $T_i$ function.

    \sdefprop*
    \begin{proof}
         Below are the proofs for both the sufficiency and necessity conditions:
        \begin{itemize}
            \item (Sufficiency) 
            
            \paragraph{General Case} Assume that some agent $i$ has the ordering $g_1 \succeq_i g_2 \dots \succeq_i g_s$ over the set of goods $S$, where some preferences may be strict. For contradiction, assume that for all $r \in [|S|]$ and $j \in N$, $\card{A_i \cap T_i(S,r)} \ge \card{A_j \cap T_i(S,r)} - 1$, but for some $g \in S$ and $j \in N$, we have $\card{A_i \cap H_i(S,g)} < \card{(A_j \setminus \set{g^j}) \cap H_i(S,g)}$, where $g^j$ is agent $i$'s most preferred good from $A_j$.

            Let $\card{H_i(S,g)} = r'$. Then from \Cref{obs:SDEF1}, we know that $T_i(S,r') = H_i(S,g)$, and from our assumption, we know that $\card{A_i \cap T_i(S,r')} \ge \card{A_j \cap T_i(S,r')} - 1$. Note that this is equivalent to saying that $\card{A_i \cap H_i(S,g)} \ge \card{A_j \cap H_i(S,g)} - 1$.

            To get the contradiction, we just need to notice that $\card{(A_j \setminus \set{g^*}) \cap H_i(S,g)} \ge \card{A_j \cap H_i(S,g)} - 1$ for all $g^* \in A_j$, since if $g^* \in H_i(S,g)$, then we have $\card{(A_j \setminus \set{g^*}) \cap H_i(S,g)} = \card{A_j \cap H_i(S,g)} - 1$, and otherwise we have $\card{(A_j \setminus \set{g^*}) \cap H_i(S,g)} = \card{A_j \cap H_i(S,g)}$.
            
            \paragraph{2 Agents} For the case of $2$ agents (Agent $i$ and Agent $j$), it is sufficient to notice that $\card{A_i \cap T_i(S,r)} \ge \floor{r/2}$ implies that $\card{A_j \cap T_i(S,r)} \le \ceil{r/2}$, since $A_j = S \setminus A_i$ when there are two agents. Since $\floor{r/2} \ge \ceil{r/2}-1$, this gives us that $\card{A_i \cap T_i(S,r)} \ge \card{A_j \cap T_i(S,r)} - 1$ for all $i,j \in N$ and $r \in [|S|]$, which we know implies SD-EF1.
            
            \paragraph{Identical Orderings} Finally, in this case, we can see that for all $i \in N$, $r \in [|S|]$, $\card{A_i \cap T_i(S,r)} \in \set{\floor{r/n},\ceil{r/n}}$ implies that $\card{A_i \cap T_i(S,r)} \ge \card{A_j \cap T_i(S,r)} - 1$. This is due to the fact that since $T_i(S,r) = T_j(S,r)$ for all $i,j$, we must have that $\card{A_i \cap T_i(S,r)} \ge \floor{r/n}$ and $\card{A_j \cap T_i(S,r)} \le \ceil{r/n}$.
    
            \item (Necessity) 
            
            \paragraph{General Case} For contradiction, assume this is false. There is some allocation $A$ over a set of goods $S$, some agent $i$ and some $r \in [|S|]$ such that $A$ is SD-EF1 and $\card{A_i \cap T_i(S,r)} < \floor{r/n}$, and for all goods $g \in T_i(S,r)$, $g' \in S \setminus T_i(S,r)$, $g \succ_i g'$. Assume that agent $i$ has the following order over the goods in $S$, $g_1 \succeq_i g_2 \dots \succeq_i g_s$ where some of the preferences may be strict.

            Let $g^*$ be the good in $T_i(S,r)$ that is not strictly preferred to any other good in $T_i(S,r)$. Since agent $i$ strictly prefers all goods in $T_i(S,r)$ to all goods that are not, we must have that $\card{H_i(S,g^*)} = r$, and thus by \Cref{obs:SDEF1}, $T_i(S,r) = H_i(S,g^*)$. Therefore, we have that $\card{A_i \cap H_i(S,g^*)} < \floor{r/n}$. 
            
            When $\card{A_i \cap H_i(S,g^*)} < \floor{r/n}$, note that by the fact that every good must be allocated to one of the $n$ agents, there must exist some $j \in N$ such that $\card{A_j \cap H_i(S,g^*)} \ge \floor{r/n} + 1$. This means that for this $j$, we have $\card{A_i \cap H_i(S,g^*)} < \card{A_j \cap H_i(S,g^*)} - 1$. But that means that for any $g^j \in A_j$, we must have $\card{A_i \cap H_i(S,g^*)} < \card{(A_j \setminus \set{g^j}) \cap H_i(S,g^*)}$, contradicting the fact that $A$ is SD-EF1.
    
            \paragraph{Identical Orderings} In the case where agents have identical orderings over the goods, it can be seen that $\card{A_i \cap T_i(S,r)} \in \set{\floor{r/n},\ceil{r/n}}$ follows as a consequence of the necessity statement in the general case. When every agent gets at least items $\floor{r/n}$ of some set of size $r$, there must be some agent $i$ who gets exactly $\floor{r/n}$ items from the set. When $n$ divides $r$, then each agent getting at least $\floor{r/n}$ goods implies that each agent gets exactly $r/n = \ceil{r/n}$ goods. Otherwise, we will have $\floor{r/n} = \ceil{r/n} - 1$, and therefore, there can be no other agent $j$ such that $\card{A_j \cap T(S,r)} > \ceil{r/n}$, or else it would immediately follow that $\card{A_i \cap H_i(S,g^*)} < \card{(A_j \setminus \set{g^j}) \cap H_i(S,g^*)}$ for some $g^*$, and all $g^j \in A_j$. 
        \end{itemize}
    \end{proof}

    \subsection{Properties of SD-PROP1}

    Here we formalize the fact that SD-EF1 implies SD-PROP1, and that SD-PROP1 implies PROP1, which completes the hereditary relationships of desiderata shown in \Cref{fig:hierarchy}.
    
    \begin{proposition}\label{prop:sdef1impliessdprop1}
        If an allocation $A$ over a set of goods $S$ is SD-EF1, then it will also be SD-PROP1.
    \end{proposition}
    \begin{proof}
        Assume that this is false, and for some allocation $A$ over a set of goods $S$, we have that $A$ is SD-EF1, but for some $i \in N$ and some $g^* \in S$, we have that $\card{(A_i \cup \set{g^i}) \cap H_i(M,g^*)} < \ceil{\nicefrac{\card{H_i(S,g^*)}}{n}}$, where $g^i$ is agent $i$'s most preferred good that is not in $A_i$.
        
        First, note that since $A$ is SD-EF1, it must be true that $\card{A_i \cap H_i(S,g^*)} \geq \floor{\nicefrac{\card{H_i(S,g^*)}}{n}}$. This is due to the fact that every good must be allocated to some agent, so if $\card{A_i \cap H_i(S,g^*)} < \floor{\nicefrac{\card{H_i(S,g^*)}}{n}}$ were true, then there must be some other agent $j \in N$ such that $\card{A_j \cap H_i(S,g^*)} \ge \floor{\nicefrac{\card{H_i(S,g^*)}}{n}} + 1$. This would directly imply that $\card{A_i \cap H_i(S,g^*)} < \card{A_j \cap H_i(S,g^*)} - 1$, so clearly $\card{A_i \cap H_i(S,g^*)} < \card{(A_j \setminus \set{g^j}) \cap H_i(S,g^*)}$ for any $g^j \in A_j$. 
        
        To see the contradiction, consider the inequality $\card{A_i \cap H_i(S,g^*)} \ge \floor{\nicefrac{\card{H_i(S,g^*)}}{n}}$. It implies that $\card{(A_i \cup \set{g^i}) \cap H_i(S,g^*)} \ge \ceil{\nicefrac{\card{H_i(S,g^*)}}{n}}$. This is because if $\card{A_i \cap H_i(S,g^*)} \le \ceil{\nicefrac{\card{H_i(S,g^*)}}{n}}$ were true, then there would need to be some good $g'$ such that $g' \in H_i(S,g^*)$ and $g' \not\in A_i$. Since $g^i$ is $i$'s most preferred good that is not in $A_i$, then we know that $g^i \succeq_i g'$, and thus $g^i \in H_i(S,g^*)$. Therefore, we would have that $\card{(A_i \cup \set{g^i}) \cap H_i(S,g^*)} \ge \floor{\nicefrac{\card{H_i(S,g^*)}}{n}} + 1 \ge \ceil{\nicefrac{\card{H_i(S,g^*)}}{n}}$
    \end{proof}

    It is worth noting that the above proof actually proves that SD-EF1 implies something slightly stronger than SD-PROP1, namely that for each $i \in N$ and $g \in S$, $\card{A_i \cap H_i(S,g)} \ge \floor{\nicefrac{\card{H_i(S,g)}}{n}}$ must hold. This can intuitively be thought of as an agent's ``round-robin share'' \citep{CFS17}, i.e., their bundle must be better than the worst possible bundle they could receive if they were to pick last in a round-robin procedure over $S$.

    \begin{proposition}
        If an allocation $A$ over a set of goods $S$ is SD-PROP1, then it will also be PROP1.
    \end{proposition}
    \begin{proof}
        Assume that some agent $i$ has the ordering $g_1 \succeq_i g_2 \dots \succeq_i g_{|S|}$ over the set of goods $S$, where some preferences may be strict, and ties are broken based on the the tie-breaking ordering of the function $T_i$. 

        Define the bundle $P_i = \set{g_1,g_{n+1},g_{2n+1},\dots}$. It must be the case that $v_i(P_i) \ge \frac{1}{n}v_i(S)$. To see this, partition $S$ into disjoint subsets of size $n$ in the form of $C_1 = \set{g_1,\dots,g_n},C_2 = \set{g_{n+1},g_{2n}},C_3 = \set{g_{2n+1},g_{3n}}$, and so on. Notice that in each subset $C$, there is a single good in $g \in P_i \cap C$, and that good is weakly preferred to all other goods in $C$. Thus for each $C$, we must have that $v_i(P_i \cap C) \ge \frac{1}{n}v_i(C)$. Summing over all subsets we get our desired inequality.

        To complete the proof, we will show that $A_i \cup \set{g^*} \succeq^\sd_i P_i$, where $g^*$ is agent $i$'s most preferred good from $S \setminus A_i$. For contradiction, assume this were false, and that for some $g \in S$, $\card{(A_i \cup \set{g^*}) \cap H_i(S,g)} < \card{P_i \cap H_i(S,g)}$. Let $\card{H_i(S,g)} = r$. By \Cref{obs:SDEF1}, we must have that $H_i(S,g) = T_i(S,r) = \set{g_1,\dots,g_r}$. Because of the way that we constructed $P_i$, we know that there cannot be more than $\ceil{\nicefrac{r}{n}} = \ceil{\nicefrac{\card{H_i(S,g)}}{n}}$ goods from $P_i$ in $\set{g_1,\dots,g_r}$. However, since we know that $A$ is SD-PROP1, it must be the case that $\card{(A_i \cup \set{g^*}) \cap H_i(S,g)} \ge \ceil{\nicefrac{\card{H_i(S,g)}}{n}}$, giving a contradiction, and showing that $A_i \cup \set{g^*} \succeq^\sd_i P_i$.
        
        $A_i \cup \set{g^*} \succeq^\sd_i P_i$ directly implies that $v_i(A_i \cup \set{g^*}) \ge v_i(P_i) \ge \frac{1}{n}v_i(S)$. Showing that $A$ is PROP1.
    \end{proof}

    Note that since SD-PROP1 only considers agents' orderings over the goods in $M$, the above proposition implies that an allocation that is SD-PROP1 for a set of orderings, is guaranteed to be PROP1 on any set of valuation function that induce those orderings.

\section{Omitted from \Cref{sec:two}}\label{app:sec-4}

\subsection{Other Fairness Desiderata that can be guaranteed by the Envy-Balancing Lemma}
In the main body of the paper, we remarked that if one can always find two allocations over each day $M_t$ that cancel out, and satisfy some fairness desiderata that implies EF1, then one can find an allocation over the entire set of goods that satisfies that desiderata per day, along with EF1 up to each day. We showed that it is always possible to find such a pair of allocations that satisfies SD-EF1. Below, we will show the same for two other interesting strengthenings of EF1.

\begin{definition}[Envy-Freeness Up to Any Good (EFX)]
    An allocation $A$ of a set of goods $S$ is \emph{envy-free up to any good} (EFX) if for all $i,j \in N$ and $g \in A_j$, $v_i(A_i) \geq v_i(A_j\setminus\set{g})$.
\end{definition}

\begin{theorem}\label{thm:2-efxper-ef1upto}
    For temporal fair division with $n=2$ agents, an allocation that is EFX per day and EF1 up to each day exists. 
\end{theorem}
\begin{proof}
    Due to \Cref{lem:envy-balancing}, it is sufficient to prove that for any day $t$, there exist EFX allocations $B_t$ and $B'_t$ of the goods in $M_t$ that cancel out. 

    We use the \cspp algorithm of \citet[Algorithm 4.2]{PR18} to produce the required two EFX allocations and show that they cancel out. Because their algorithm has a simpler description for additive valuations (which we focus on), we explicitly describe the construction here. Allocation $B_t$ is constructed as follows. 
    \begin{enumerate}
        \item Find a partition $(P,Q)$ of $M_t$ that minimizes $|v_1(P) - v_1(Q)|$. With loss of generality, assume that $v_1(P) \geq v_1(Q)$. 
        \item If there are any goods $g \in P$ such that $v_1(g) = 0$, move them to $Q$. Note that this does not change $v_1(P)$ or $v_1(Q)$. 
        \item Allow agent $2$ to pick their preferred bundle, and assign the other bundle to agent $1$.
    \end{enumerate}
    Allocation $B'_t$ is computed similarly, but reversing the roles of agents $1$ (who picks) and agent $2$ (who cuts). 
    
    \citet[Theorem 4.3]{PR18} show that the resulting allocations, $B_t$ and $B'_t$ are EFX. It remains to show that they cancel out. Due to symmetry, we simply need to argue the cancellation for agent $1$, i.e.,
    \[
    v_1(B_{t,1}) + v_1(B'_{t,1}) \ge v_1(B_{t,2}) + v_1(B'_{t,2}) \Leftrightarrow v_1(B'_{t,1})-v_1(B'_{t,2}) \ge v_1(B_{t,2})-v_1(B_{t,1}).
    \]
    Note that $v_1(B'_{t,1})-v_1(B'_{t,2}) \ge 0$ because agent $1$ picks their favorite out of the two bundles in $B'_t$. Further, $|v_1(B'_{t,1})-v_1(B'_{t,2})| \ge |v_1(B_{t,1})-v_1(B_{t,2})|$ because $(B_{t,1},B_{t,2})$ is the partition that minimizes the difference between agent $1$'s value for the two bundles. Putting the two together, we get the desired result. 
\end{proof}

Next, we derive the same result for EF1+PO. Interestingly, we use the two cancelling EFX allocations produced for the previous result in order to show the existence of two cancelling EF1+PO allocations. This requires adding a minor tie-breaking rule to the procedure for computing these EFX allocations (the \cspp algorithm due to \citet{PR18}): when the chooser (agent $2$ in $B_t$) is indifferent between the two bundles but the cutter (agent $1$ in $B_t$) is not, the chooser must pick the bundle the cutter values less. An interested reader can note that the allocations would have remained EFX even if we had introduced this tie-breaking in the proof of \Cref{thm:2-efxper-ef1upto}, but it was not needed there. 

\begin{definition}[Pareto Optimality (PO)]
    An allocation $A$ of a set of goods $S$ is \emph{Pareto optimal} (PO) if there is no allocation $A'$ such that $v_i(A'_i) \geq v_i(A_i)$ for all $i \in N$ and at least one inequality is strict.
\end{definition}

\begin{theorem}\label{thm:2-ef1poper-ef1upto}
    For temporal fair division with $n=2$ agents, an allocation that is EF1+PO per day and EF1 up to each day exists. 
\end{theorem}
\begin{proof}
    Due to \Cref{lem:envy-balancing}, it is sufficient to prove that for any day $t$, there exist EF1+PO allocations $B_t$ and $B'_t$ of the goods in $M_t$ that cancel out. 

    We claim that for any day $t$, the existence of $2$ EFX allocations that cancel out implies the existence of $2$ EF1+PO allocations that cancel out. For contradiction, assume this were false, and for some instance there are not $2$ EF1+PO allocations that cancel out.

    Let the allocation $B_t$ be the allocation that was constructed by the process in \Cref{thm:2-efxper-ef1upto}, where the bundles are chosen according to agent $1$'s valuations, and agent $2$ picks their preferred bundle. We know that in this allocation, $v_2(B_{t,2}) \geq v_2(B_{t,1})$, and for all $g \in B_{t,2}, v_1(B_{t,1}) \geq v_1(B_{t,2}\setminus \{g\})$. If $B_t$ is a PO allocation, then we do not have to deal with it further, otherwise, we know that there exists some $P \subseteq B_{t,1}, Q \subseteq B_{t,2}$, such that the allocation $((B_{t,1}\setminus P)\cup Q, (B_{t,2}\setminus Q)\cup P)$ is PO, and that $v_1((B_{t,1}\setminus P)\cup Q) \geq v_1(B_{t,1})$ and $v_2((B_{t,2}\setminus Q)\cup P) \geq v_2(B_{t,2})$.
    
    First, note that in this PO allocation $Q$ must be a strict subset of $B_{t,2}$. If this were not true, then $v_2((B_{t,2}\setminus Q)\cup P) \geq v_2(B_{t,2})$ would tell us that $v_2(B_{t,1}) \ge v_2(P) \ge v_2(B_{t,2})$, with the first inequality being due to the fact that $P \subseteq B_{t,1}$. By the process used to construct $B_t$, we already know that $v_2(B_{t,2}) \ge v_2(B_{t,1})$, meaning $v_2(B_{t,2}) = v_2(B_{t,1})$, and by the tie-breaking mechanism in the cut-and-choose algorithm, this would mean that $((B_{t,1}\setminus P)\cup Q, (B_{t,2}\setminus Q)\cup P)$ must be an EF+PO allocation. This would lead to a contradiction, since $2$ copies of any EF allocation clearly cancel out.
    
    Therefore, we must have that $Q \subset B_{t,2}$. Due to this, we know that there will be some good $g$ such that $g \in (B_{t,2}\setminus Q) \cup P$ and $g \in B_{t,2}$. After the reallocation, we know that no agent was made worse off, so it will be the case that Agent $2$ will not feel envy in the new PO allocation. We know that $((B_{t,1}\setminus P)\cup Q, (B_{t,2}\setminus Q)\cup P)$ cannot be an EF allocation (by the logic from the above paragraph this would lead to contradiction), so we know that Agent $1$ does feel envy in the PO allocation. Since $B_t$ is EFX, we know that for all $g' \in B_{t,2}$, we have that $v_1(B_{t,1}) \ge v_1(B_{t,2}) - v_1(g')$. Since we only have $2$ agents, we have that $v_1(B_{t,2}) = v_1(M_t) - v_1(B_{t,1})$, and thus $v_1(g') \ge v_1(M_t) - 2v_1(B_{t,1})$ for all $g' \in B_{t,2}$. Therefore, due to the fact that $v_1((B_{t,1}\setminus P)\cup Q) \geq v_1(B_{t,1})$, we have that $v_1(g) \ge v_1(M') - 2v_1((B_{t,1}\setminus P)\cup Q)$, which can be rearranged back into $v_1((B_{t,1}\setminus P)\cup Q) \ge v_1((B_{t,2}\setminus Q)\cup P) - v_1(g)$, showing that this PO allocation is also EF1.
    
    We can symmetrically repeat this procedure with the other EFX allocation $B'_t$ where the bundles are selected according to Agent $2$'s valuations and Agent $1$ chooses their preferred bundle. Since our PO reallocation can only increase the utility of both agents (and thus lessen their envy), we can conclude that if the original EFX allocations cancel out, the corresponding EF1+PO allocations also cancel out, giving a contradiction.
\end{proof}

We note that while we provided a polynomial-time algorithm in \Cref{thm:2-sdef1per-ef1each} for achieving SD-EF1 per day and EF1 up to each day result, the constructions in \Cref{thm:2-efxper-ef1upto,thm:2-ef1poper-ef1upto} for achieving EFX or EF1+PO per day are not efficient because they rely on partitioning a set of numbers into two subsets with near-equal sum. This is NP-hard because PARTITION (which requires exactly equal sum) can be trivially reduced to it. This raises the following interesting open question:

\begin{open}
    For temporal fair division with $n=2$ agents, can EFX or EF1+PO per day and EF1 up to each day be achieved in polynomial time?
\end{open}

\subsection{Two Agents and Identical Days}\label{sec:two-days}
In this section, we briefly discuss the implications of our results for the special case where we have two agents and identical days. Recall that \citet{igarashi2024repeated} focus only on the case of identical days and many of their results hold for only two agents. In particular, they show that an allocation that is EF1 per day and (exact) EF overall exists and can be computed in polynomial time when the number of days $k$ is even. 

We prove a slightly stronger result via a much simpler technique, albeit only for allocating goods while their result is for allocating a mixture of goods and chores.

\begin{proof}[Proof of \Cref{2-iddays-SDEF1}]
    Consider the set of goods $M_1$ on day $1$ (each day has a set of goods identical to this). Consider the two allocations $B = (B_1,B_2)$ and $B' = (B_2,B_1)$ produced in \Cref{lem:2-opposite-sdef1} such that both $B$ and $B'$ are SD-EF1 allocations of $M_1$. The desired allocation is one that uses $B$ on every odd day and $B'$ on every even day. Clearly, SD-EF1 per day is satisfied. Since the allocations completely cancel out after every even day (each agent has exactly the same number of copies of each good), we get SD-EF1 up to each day and SD-EF up to each even day.
\end{proof}

\section{Omitted From \Cref{sec:laminar}}\label{app:sec-7}

\subsection{Laminar Set Families}

For any set of goods $M$, a collection $\mathcal{L}$ of subsets of $M$ is a laminar set family over $M$ if for every pair of subsets $S,T \in \mathcal{L}$, we have that either $S \cap T = \emptyset$, $S \subset T$, or $T \subset S$.
    
For each set $S \in \mathcal{L}$, let $D(S): \mathcal{L} \rightarrow 2^\mathcal{L}$ be a function returning each set $S' \in \mathcal{L}$ such that $S' \subset S$. Let $C(S): \mathcal{L} \rightarrow 2^\mathcal{L}$ be the function returning every maximal set of $D(S)$ (every set in $D(S)$ that is not contained in some other set from $D(S)$).

We can say that a laminar set family $\mathcal{L}$ over $M$ is \emph{Complete} if and only if the following conditions hold:
\begin{itemize}
    \item The set $M$ is an element of $\mathcal{L}$.
    \item For every set $S \in \mathcal{L}$, either $D(S) = \emptyset$, or $\cup_{S' \in D(S)} = S$.
\end{itemize}

We will assume that all laminar sets families we deal with in our setting are complete. This can be assumed without loss of generality, as any laminar set family $\mathcal{L}$ over a set of goods $M$ that is not complete can be ``completed'' through the following simple procedure:

\begin{itemize}
    \item If $M \not\in \mathcal{L}$, add $M$ to $\mathcal{L}$.
    \item For every set $S \in \mathcal{L}$, if $0 < |\cup_{S' \in D(S)}|< |S|$, then create a new set $S^* = S \setminus \cup_{S' \in D(S)}$ and add it to $\mathcal{L}$.
\end{itemize}

After the above steps, $\mathcal{L}$ will remain a laminar set family. $\mathcal{L}$ will clearly remain laminar after the addition of $M$ since every other set in $\mathcal{L}$ will be a strict subset of $M$. $\mathcal{L}$ will also remain laminar after the addition of each of the $S^*$ sets, since the definition of laminar families ensures that if none of the sets in $\mathcal{L}$ that are subsets of a set $S$ contain some good $g$, then the only other sets in $\mathcal{L}$ that can contain $g$ are strict supersets of $S$. Therefore, the $S^*$ that the completion process adds to $\mathcal{L}$ will be a strict subset of $S$, thus a strict subset of all supersets of $S$, and will have no overlap with any other sets.

Note that in any complete laminar set family $\mathcal{L}$, and any $S \in \mathcal{L}$, $C(S)$ will either be empty, or will form a complete disjoint partitioning of $S$. If this were not true, then we would have that $\cup_{S' \in C(S)} \subset S$ and $\cup_{S' \in D(S)} = S$. Clearly, every good $g \in \cup_{S' \in D(S)}$ must appear in some maximal set from $D(S)$, giving a contradiction. The fact that all the sets in $C(S)$ will be disjoint follows immediately from the definition of a laminar set family and from the fact that each set $S' \in C(S)$ is maximal in $D(S)$.

One can think of the structure of a complete laminar set family $\mathcal{L}$ as a tree where each subset is a node. $M$ is the root node of the tree. For any set $S \in \mathcal{L}$, $D(S)$ are the \textbf{descendants} on $S$, and $C(S)$ are the \textbf{children} of $S$. The \textbf{Leaf sets} of $\mathcal{L}$ are any sets $S \in \mathcal{L}$ who have no descendants (i.e. $D(S) = \emptyset$). Thinking of complete Laminar Set Families in this way will allow us to topologically sort the sets. Particularly, thinking of $\mathcal{L}$ as a directed graph where there is a directed edge going from each child to its parent, then a topological sorting of $\mathcal{L}$ will result in no set appearing in the ordering before any of its descendants.

\subsection{Laminar EF1}

With the definitions of laminar set families in mind, we can now prove \Cref{lem:laminar-envy-balancing}.

\begin{algorithm}
    \caption{EnvyBalancing++}
    \textbf{Input} {A fair division instance $(N,M,v)$; A partition $C = \set{C_1,\dots,C_k}$ over the goods in $M$; For $t \in [k]$, a pair of allocations $(B^1_t,B^2_t)$ of the set of goods $C_t$ that cancel out, with labels assigned in such a way that if neither of the allocations are Envy-Free, then $v_1(B^1_{t,1}) - v_1(B^1_{t,2}) \geq 0$ and $v_2(B^2_{t,2}) - v_2(B^2_{t,1}) \geq 0$.}
    
	\textbf{Output} {An allocation $A$ of the set of goods $M$}
    \begin{algorithmic}[1]
        \IF{$C = \emptyset$}
            \STATE $(A,A') \gets $ Two EF1 allocations of $S$ that cancel out (chosen by any subroutine)
        \ELSE
            \STATE $F \gets \emptyset$, $S \gets \emptyset$, $e_1 \gets 0$, $e_2 \gets 0$
            \FOR{$t \in [k]$}
                \IF{$v_1(B^1_{t,1}) \geq v_1(B^1_{t,2}) \land v_2(B^1_{t,2}) \geq v_2(B^1_{t,1})$}
                    \STATE $F \leftarrow F \cup \{B_t\}$
                \ELSIF{$v_1(B^2_{t,1}) \geq v_1(B^2_{t,2}) \land v_2(B^2_{t,2}) \geq v_2(B^2_{t,1})$}
                    \STATE $F \leftarrow F \cup \{B'_t\}$
                \ELSE
                    \IF{$e_1 \leq 0$}
                        \STATE $S \leftarrow S \cup \{B^1_t\}$
                        \STATE $e_1 \leftarrow e_1 + (v_1(B^1_{t,1}) - v_1(B^1_{t,2}))$
                        \STATE $e_2 \leftarrow e_2 + (v_2(B^1_{t,2}) - v_2(B^1_{t,1}))$
                    \ELSE
                        \STATE $S \leftarrow S \cup \{B^2_t\}$
                        \STATE $e_1 \leftarrow e_1 + (v_1(B^2_{t,1}) - v_1(B^2_{t,2}))$
                        \STATE $e_2 \leftarrow e_2 + (v_2(B^2_{t,2}) - v_2(B^2_{t,1}))$
                    \ENDIF
                \ENDIF
    
                \IF{$e_1 \geq 0 \land e_2 \geq 0$}
                    \STATE $F \leftarrow F \cup S$
                    \STATE $S \leftarrow \emptyset$, $e_1 \leftarrow 0$, $e_2 \leftarrow 0$
                \ELSIF{$e_1 \leq 0 \land e_2 \leq 0$}
                    \STATE $F \leftarrow F \cup \text{SWAP}(S)$
                    \STATE $S \leftarrow \emptyset$, $e_1 \leftarrow 0$, $e_2 \leftarrow 0$
                \ENDIF
            \ENDFOR
            \STATE $A \gets$ allocation in which each $C_t$ is allocated according to the allocation of $C_t$ in $F \cup S$, for each $t \in [k]$
            \STATE $A' \gets$ allocation in which $C_t$ is allocated according to the allocation of $C_t$ in $F \cup \text{SWAP}(S)$, for each $t \in [k]$ 
        \ENDIF
        \RETURN $(A,A')$
    \end{algorithmic}
    \label{alg:envy-balancing++}
\end{algorithm}

\begin{algorithm}
    \caption{Laminar Envy-Balancing Algorithm}
    \textbf{Input} {A fair division instance $(N,M,v)$; A laminar set $\mathcal{L}$ over $M$}
    
	\textbf{Output} {An allocation $A$ of the set of goods $M$}
    \begin{algorithmic}[1]
        \STATE $T \gets (S_1,S_2,\dots,S_{|\mathcal{L}|})$ (A topological ordering of the sets in $\mathcal{L}$)
        \STATE $\forall S \in \mathcal{L}, B_S \gets \emptyset$
        \FOR{$t \in [|\mathcal{L}|]$}
            \STATE $(B^1_S, B^2_S) \gets $ EnvyBalancing++($(N,S_t,v)$, $C(S_t)$, $\set{B_{S'} : S' \in C(S_t)}$)
            \STATE $B_{S_t} = (B^1_S, B^2_S)$
        \ENDFOR
        \RETURN $B_{S_{|\mathcal{L}|}}$
    \end{algorithmic}
    \label{alg:laminar-envy-balancing}
\end{algorithm}

\laminarenvybalancing*
\begin{proof}
    We begin by introducing a slight adaptation of the Envy-Balancing algorithm from \Cref{lem:envy-balancing}, which will serve as the basis for our algorithm to give the stronger guarantee of laminar fairness. We will use the notation introduced in the proof of \Cref{lem:envy-balancing} to describe the behaviour of the algorithm.

    At the end of the \Cref{alg:envy-balancing}, the allocation $A$ induced by $F \cup S$ is returned as the final allocation. $A$ is an EF1 allocation since we know that $A_F$ is an EF allocation and $A_S$ is EF1. Note that $A_{F \cup \text{SWAP}(S)}$ would also be an EF1 allocation for the same reason (The proof of \Cref{lem:envy-balancing} shows that at every point in \Cref{alg:envy-balancing}, both $A_S$ and $A_{\text{SWAP}(S)}$ are EF1 allocations). Also note that the two allocations $A_{F \cup S}$ and $A_{F \cup \text{SWAP}(S)}$ will cancel out. This is because we know that neither agent feels envy in $A_F$, and we also know that $A_S$ and $A_{\text{SWAP}(S)}$ will cancel out.

    This was shown in the proof for \Cref{lem:envy-balancing}, for completeness we will show it here as well. For every day $t$, let the inputs to the Envy-Balancing algorithm for day $t$ be $(B^1_t,B^2_t)$, where if neither is EF, $B^1_t$ is agent $1$'s preferred allocation, and $B^2_t$ is agent $2$'s preferred allocation. Let $D \subseteq [k]$ be the set of days such that an allocation for day $t$ appears in $S$. From the fact that we know the pair of allocations $(B^i_t,B^{3-i}_t)$ on each day cancels out, we have:

    \[\sum_{t \in D}{(v_i(B^i_{t,i}) + v_i(B^{3-i}_{t,i}))} \geq \sum_{t \in D}{(v_i(B^i_{t,3-i}) + v_i(B^{3-i}_{t,3-i}))}\]

    Then note that we have $v_i(A_{S,i}) + v_i(A_{\text{SWAP}(S),i}) = \sum_{t \in D}{(v_i(B^i_{t,i}) + v_i(B^{3-i}_{t,i}))}$, and $v_i(A_{S,3-i}) + v_i(A_{\text{SWAP}(S),3-i}) = \sum_{t \in D}{(v_i(B^i_{t,3-i}) + v_i(B^{3-i}_{t,3-i}))}$. This is because for each $t \in D$, both $B^i_t$ and $B^{3-i}_t$ are contained in exactly one of $S$ or SWAP$(S)$. This implies that $A_S$ and $A_{\text{SWAP(S)}}$ cancel out.

    We will use this fact in order to construct a new algorithm, \Cref{alg:envy-balancing++}, named EnvyBalancing++. This algorithm takes as input a set of goods $S$ and a partitioning $C$ over the goods in $S$. If $C = \emptyset$, then EnvyBalancing++ simply returns $2$ EF1 allocations over $S$ that cancel out, using any arbitrary method to do so (such as the one introduced in \Cref{lem:2-opposite-sdef1}). Otherwise, if $C$ is a complete and disjoint partition of $S$, then EnvyBalancing++ runs \Cref{alg:envy-balancing} on the input, but returns both $A_{F \cup S}$ and $A_{F \cup \text{SWAP}(S)}$ as the final allocations.
    
    \Cref{alg:laminar-envy-balancing} uses EnvyBalancing++ as a subroutine, and is the algorithm that will find the desired laminar fair allocation. It takes in a laminar set family $\mathcal{L}$, and finds an allocation that is laminar EF1 with respect to $\mathcal{L}$ by the following process:

    \begin{itemize}
        \item Sort the sets from $\mathcal{L}$ topologically such that no set appears in the order before any of its descendants.
        \item Run EnvyBalancing++ on each set $S \in \mathcal{L}$ in the topological order, with the partitioning of $S$ being given by $C(S)$, and the pairs of allocations for each subset in $C(S)$ being the allocations the algorithm has already found for each of the children of $S$. If $S$ is a leaf set, then $C(S) = \emptyset$, and the algorithm will return two arbitrary EF1 allocations over $S$ that cancel out. Otherwise, two allocations generated by the Envy-Balancing algorithm will be returned.
        \item As the final output, the algorithm returns the $2$ allocations over $M$, which will be the last set from $\mathcal{L}$ ordered topologically. Each of these allocations are guaranteed to be EF1 with respect to every set $S \in \mathcal{L}$.
    \end{itemize}

    We will prove the following inductive statement: For any set $S \in \mathcal{L}$, if the algorithm found $2$ allocations for each child $S' \in C(S)$ that are EF1 with respect to $S'$ and all the descendants of $S'$, then the algorithm will find $2$ allocations over $S$ that are EF1 with respect to $S$ and all descendants of $S$.

    We will first start with the case where $S$ is not a leaf set. Since the algorithm visits each descendant of $S$ prior to visiting $S$, if the hypothesis holds, then it will have already found $2$ allocations for every set $S' \in C(S)$ which will cancel out, will be EF1 with respect to $S'$, and will be EF1 with respect to every descendent of $S'$. The algorithm will use these allocations as input to EnvyBalancing++. Since we are assuming that $\mathcal{L}$ is a complete laminar set family, the children of $S$ will make a complete and disjoint partition of $S$. We know that EnvyBalancing++ will output two allocations over $S$ that cancel out and are EF1 over $S$. Further, we know that both allocations will be EF1 with respect to all descendants of $S$ due to the fact that the fact that the final allocation outputted by EnvyBalancing++ will be induced by picking one of the inputted allocations from each of child of $S$, which are all known to be EF1 with respect to the child and all its descendants from the hypothesis.

    As a base case, if $S$ is a Leaf set, the algorithm uses any arbitrary method to find the two allocations for $S$. These allocations will be EF1 with respect to $S$, and since leaf sets do not have any descendants, they will vacuously be EF1 with respect to all descendants as well. Thus, both of the final allocations returned by the algorithm will be EF1 with respect to $M$ and all of its descendants. Since $D(M) = \mathcal{L} \setminus M$, these allocations with be Laminar EF1 with respect to $\mathcal{L}$.

    Finally, topological sorting of the sets will take polynomial-time in the number of elements in the tree it is traversing, and it is well-known that every laminar set family with a ground set $M$ can have at most $2|M|-1$ members. This together with the fact that EnvyBalancing++ clearly runs in polynomial time allows us to conclude that the entire procedure will be polynomial.
\end{proof}
\end{document}